\documentclass[11pt,a4paper,pdftex]{amsart}
\usepackage{framed}
\usepackage{float}
\usepackage{natbib}
\usepackage{mathrsfs}

\usepackage{mathtools}
\usepackage{url}
\usepackage{comment}
\usepackage{subcaption}
\usepackage{bbm}
\usepackage{stmaryrd}
\usepackage[pdftex]{hyperref,color,graphicx}
\newtheorem{theorem}{Theorem}[section]

\newtheorem{lemma}[theorem]{Lemma}

\newtheorem{proposition}[theorem]{Proposition}

\theoremstyle{definition}

\theoremstyle{remark}
\newtheorem{assumption}[theorem]{Assumption}
\newtheorem{remark}[theorem]{Remark}
\newtheorem{example}[theorem]{Example}

\numberwithin{equation}{section}

\newcommand{\conv}{\mathrm{conv}}
\newcommand{\cl}{\mathrm{cl}}

\newcommand{\ud}{\,\mathrm{d}}

\newcommand{\e}{\mathrm{e}}
\newcommand{\Exp}{\mathrm{Exp}}
\begin{document}
%
%
%
%
\title[Randomly terminating income]{Duality for optimal consumption
  with randomly terminating income}

\author[Ashley Davey]{Ashley Davey}

\address{Ashley Davey \\
Department of Mathematics \\
Imperial College London \\
London SW7 2BZ \\
UK}

\email{ashley.davey18@imperial.ac.uk}

\author[Michael Monoyios]{Michael Monoyios}

\address{Michael Monoyios \\
Mathematical Institute \\ 
University of Oxford \\
Radcliffe Observatory Quarter \\
Woodstock Road\\ 
Oxford OX2 6GG \\
UK}

\email{monoyios@maths.ox.ac.uk}

\author[Harry Zheng]{Harry Zheng}

\address{Harry Zheng \\
Department of Mathematics \\
Imperial College London \\
London SW7 2BZ \\
UK}

\email{h.zheng@imperial.ac.uk}

\date{\today}

\maketitle



{\centering\footnotesize Dedicated to the memory of Mark
H. A. Davis \par}

\begin{abstract}

  We establish a rigorous duality theory, under No Unbounded Profit
  with Bounded Risk, for an infinite horizon problem of optimal
  consumption in the presence of an income stream that can terminate
  randomly at an exponentially distributed time, independent of the
  asset prices. We thus close a duality gap encountered by \cite{dv09}
  in a version of this problem in a Black-Scholes market. Many of the
  classical tenets of duality theory hold, with the notable exception
  that marginal utility at zero initial wealth is \textit{finite}. We
  use as dual variables a class of supermartingale deflators such that
  deflated wealth plus cumulative deflated consumption in excess of
  income is a supermartingale. We show that the space of discounted
  local martingale deflators is dense in our dual domain, so that the
  dual problem can also be expressed as an infimum over the discounted
  local martingale deflators. We characterise the optimal wealth
  process, showing that optimal deflated wealth is a potential
  decaying to zero, while deflated wealth plus cumulative deflated
  consumption over income is a uniformly integrable martingale at the
  optimum. We apply the analysis to the \cite{dv09} example and give a
  numerical solution.

\end{abstract}

\maketitle



\noindent \textit{MSC 2010 subject classifications:} {93E20, 91G80,
  91G10, 49M29, 49J55, 49K45}

\vspace{0.2cm}

\noindent \textit{Keywords and phrases:} {Duality, utility from
consumption, portfolio optimisation, supermartingale deflator,
terminating income, HJB equation}


\section{Introduction}
\label{sec:intro}

In this paper we establish a duality theory for an infinite horizon
optimal consumption problem in which the agent also receives an income
stream which terminates at a random exponentially distributed time,
independent of the filtration governing the evolution of the asset
prices. Such a problem was considered by \cite{dv09} in a
Black-Scholes market. Here, we consider a general semimartingale
incomplete market, under the minimal no-arbitrage assumption of no
unbounded profit with bounded risk (NUPBR). We thus assume only the
existence of a suitable class of deflators, and use no arguments
involving equivalent local martingale measures (ELMMs). This is
natural under NUPBR, and also desirable in a perpetual model, since
ELMMs do not typically exist over infinite horizons.

In \cite{dv09}, despite the apparent simplicity of the modification of
the classical Merton problem with deterministic (or indeed no) income,
the problem proved a remarkably intractable one to solve and
understand. \cite{dv09} implemented a differential equation-based dual
approach, made an ansatz that the dual control was deterministic, and
used differential equation heuristics to argue that the derivative of
the value function at zero initial wealth was infinite. Ultimately,
though, they encountered a duality gap, in that their derived value
function could not solve the Hamilton-Jacobi-Bellman (HJB) equation,
indicating that the dual optimiser must indeed be stochastic and
state-dependent in some way. What is more, their numerical solutions
indicated that the value function derivative at zero initial wealth
was \textit{finite}. The intuition behind this last feature is clear:
even though the income can terminate very soon after time zero, it
cannot terminate immediately, except in the limiting case that the
intensity of the exponential time approaches infinity. Thus, the agent
is bound to receive some income, so is not infinitely penalised for
having zero initial capital, and marginal utility of wealth is finite
at this point.

The above issues make optimal consumption with randomly terminating
income an open and interesting problem. We provide a dual
characterisation of the problem in a general set-up, and close the
duality gap. We find that most of the usual tenets of duality theory
hold, with the exception that the marginal utility at zero initial
wealth is indeed finite. We also characterise the optimal wealth
process, and show that the optimal deflated wealth process is a
potential, decaying almost surely to zero, while deflated wealth plus
cumulative deflated consumption over income is a uniformly integrable
martingale at the optimum. Our results apply also to the finite
horizon version of the problem, both with and without a terminal
wealth objective, and we describe the minor adjustments needed to do
this in Section \ref{subsec:finiteh}.

Our analysis is based on characterising the dual domain by a
fundamental supermartingale property, using a class of deflators
(\textit{consumption deflators}) such that deflated wealth plus
cumulative deflated consumption in excess of income is a
supermartingale. The supermartingale property yields a budget
constraint as a necessary condition (Lemma \ref{lem:bc}) for
admissible consumption plans. Motivated by the budget constraint, we
assume a financing condition (Assumption \ref{ass:financing}) that
characterisies admissibility, so that the budget constraint is also a
sufficient condition for admissibility. This ensures that the primal
domain is closed in an appropriate topology.
We show that our dual domain is closed and coincides with the closure
(in an appropriate topology) of the discounted local martingale
deflators (LMDs). But we are able to avoid having to explicitly invoke
such a closure in defining our dual domain, which yields a stronger
duality statement. This is discussed in Remark \ref{rem:cdvslmd}.

The mathematical contribution of the paper is to show how the seminal
approach to convex duality for utility maximisation, as inspired by
\cite{ks99,ks03} and recently adapted to the no-income infinite
horizon consumption problem by \cite{mmc20}, can be adapted to a
problem with random endowment. The significant adaptations required by
the presence of the income are that, first, the bipolar theorem of
\cite{bs99} is unavailable, becuase the primal variable appearing in
the budget constraint is the difference between the consumption and
income rates, so is not guaranteed to be non-negative, precluding an
application of the bipolar theorem. Second, we do not enlarge the
original dual domain of deflators to encompass all processes dominated
by the deflators. This enlargement is typically carried out to reach
the bipolar of the original dual domain. Not only is this not of use
here, since the bipolar theorem is inapplicable, but the the presence
of the random endowment renders the dual objective non-monotone in the
dual variables, so the enlargement could result in the dual minimiser
not coinciding with the minimiser in the enlarged domain. This is
discussed in Remark \ref{rem:solid}.

The upshot of these features is that the important bipolarity result
\cite[Proposition 3.1]{ks99} that underpins classical duality has to
be replaced by a suitable analogue. This is provided by Proposition
\ref{prop:propertiesCD}, giving the key properties of the primal and
dual domains in the abstract version of the problem. In particular,
Proposition \ref{prop:propertiesCD} states that the dual domain is
closed, and bounded in $L^{1}(\mu)$ (with $\mu$ the measure in the
abstract formulation of the problem). In the terminal wealth problem
of \cite{ks99,ks03}, the dual domain being bounded in
$L^{1}(\mathbb{P})$ follows simply from the fact that the unit wealth
process is admissible. Here, the argument is more subtle, and requires
a strictly positive interest rate, which is natural in a perpetual
consumption problem with income.

With Proposition \ref{prop:propertiesCD} in place, some of the
classical steps to a duality theorem can be brought to bear on the
problem. We prove an abstract duality theorem (Theorem \ref{thm:adt}),
from which the concrete duality (Theorem \ref{thm:cwrtid}) follows,
including a characterisation of the optimal wealth process, as well as
the property of finite marginal utility at zero wealth. Note that
other papers on utility from consumption with random endowment are not
able to cover our model or to produce the corresponding results. In
some papers, such as \cite{kz03} (or \cite{csw01}, with a terminal
wealth objective), the models are finite horizon, under the No Free
Lunch with Vanishing Risk (NFLVR) no-arbitrage assumption, so heavily
reliant on ELMMs and the dual space is required to incorporate a
singular part, in the form of finitely additive measures, so a unique
dual optimiser of the form we obtain is elusive. In other papers, such
as \cite{hk04,most17,mostsirbu20}, finitely additive measures are
avoided by expanding the dimension of the value function to
incorporate an additional variable describing the number of units of
the random endowment. These works are again over a finite horizon
under NFLVR, so also reliant on ELMMs, and the expansion of dimension
of the value function renders them unable to yield the sharp
differentiability results we obtain for the value functions, only
super-and sub-differential results.  Some early papers on optimal
investment and consumption with random endowment, usually in a
Brownian filtration, adopt a HJB equation perspective, such as
\cite{dz93} and \cite{hp93}, and do not cover our general model.

Finally, we analyse  the \cite{dv09} example, exploring the
ramifications of duality, and numerically solving the problem. By
integrating over the distribution of the random income termination
time, we transform the pre-income termination problem to a stochastic
control problem with perpetual income in which utility is derived from
consumption and inter-temporal wealth, and numerically solve the HJB
equation for the resulting problem.



The remainder of the paper is structured as follows. In Section
\ref{sec:probform} we formulate the primal problem, budget constraint
and the dual problem and outline the \cite{dv09} example. In Section
\ref{sec:tmd} we state the main duality theorem (Theorem
\ref{thm:cwrtid}) and also Theorem \ref{thm:dense}, that the infima
over consumption deflators and discounted local martingale deflators
coincide. In Section \ref{sec:tad} we state key properties of the
abstract primal and dual domains (Proposition
\ref{prop:propertiesCD}), the abstract duality theorem (Theorem
\ref{thm:adt}) and the result (Proposition \ref{prop:dense}) which
underlies Theorem \ref{thm:dense}. The proofs of these results are
given in Section \ref{sec:pdts}. In Section \ref{sec:dve} we examine
the \cite{dv09} example. Section \ref{sec:conclusion} concludes.

\section{Problem formulation}
\label{sec:probform}

\subsection{The financial market}
\label{subsec:market}

We consider an infinite horizon investment-consumption problem in a
semimartingale incomplete market, and in the presence of an income
stream which terminates at an exponentially distributed time that is
independent of the filtration governing the asset prices. We have a
complete stochastic basis
$(\Omega,\mathcal{F},\mathbb{F}:=(\mathcal{F}_{t})_{t\geq 0},
\mathbb{P})$, with the filtration $\mathbb{F}$ satisfying the usual
hypotheses of right-continuity and augmentation with $\mathbb{P}$-null
sets of $\mathcal{F}$. The filtration $\mathbb{F}$ is given by
\begin{equation*}
\mathbb{F} := \mathbb{G}\vee\mathbb{F}^{N},  
\end{equation*}
where $\mathbb{G}$ is the ($\mathbb{P}$-augmentation) of the
filtration governing the evolution of a $d$-dimensional non-negative
semimartingale stock price vector $S=(S^{1},\ldots,S^{d})$ and
$\mathbb{F}^{N}$ is a filtration independent of $\mathbb{G}$, and is
the augmentation of the filtration generated by the c\`adl\`ag process
$N$, given by
\begin{equation}
N_{t} := \mathbbm{1}_{\{t<\tau\}}, \quad t\geq 0,   
\label{eq:N}
\end{equation}
where $\tau\sim\mathrm{Exp}(\eta)$ is an exponentially distributed
time with parameter $\eta\geq 0$, independent of $\mathbb{G}$. The
interest rate is a strictly positive $\mathbb{G}$-adapted process
$r=(r_{t})_{t\geq 0}$ satisfying
$\int_{0}^{t}r_{s}\ud s<\infty,\,t\geq 0$. We assume that the interest
rate is bounded below by a positive constant,
\begin{equation}
r_{t} \geq \underline{r} > 0, \quad \forall \, t\geq 0.
\label{eq:rbound} 
\end{equation}

An agent with initial capital $x\in\mathbb R_{+}$ trades the stocks
plus cash, consumes wealth at a non-negative adapted rate
$c=(c_{t})_{t\geq 0}$, and receives income at some non-negative
bounded adapted rate $f=(f_{t})_{t\geq 0}$. The consumption rate is
assumed to almost surely satisfy the minimal integrability condition
$\int_{0}^{t}c_{s}\ud s<\infty,\,t\geq 0$.

The income stream is specified as follows. With $a=(a_{t})_{t\geq 0}$
a non-negative $\mathbb{G}$-adapted process, bounded above by some
constant $\overline{a}>0$, we have a \textit{randomly terminating
  income stream}, with $f$ given by
\begin{equation}
f_{t} := a_{t}N_{t} = a_{t}\mathbbm{1}_{\{t<\tau\}}, \quad 0\leq
a_{t}\leq \overline{a}<\infty, \quad t\geq 0.
\label{eq:fbound}
\end{equation}
Thus, the stochastic income stream pays at the bounded
$\mathbb{G}$-adapted rate $a$ up to the random time $\tau$, at which
point it abruptly terminates. The random termination of the income
generates an additional source of market incompleteness above and
beyond any inherent incompleteness in the original market in the
absence of the income stream.

The agent's trading strategy $H=(H^{1},\ldots,H^{d})$ is a predictable
$S$-integrable process for the number of shares of each stock
held. The agent's wealth process, $X$, follows
\begin{equation}
\ud X_{t} = H_{t}\ud S_{t} + r_{t}\left(X_{t} - H_{t}S_{t}\right)\ud t
- c_{t}\ud t + f_{t}\ud t, \quad X_{0}=x>0, 
\label{eq:wealth}
\end{equation}
where for brevity we write
$H_{t}\ud S_{t}\equiv\sum_{i=1}^{d}H^{i}_{t}\ud S^{i}_{t}$ and
$H_{t}S_{t}\equiv\sum_{i=1}^{d}H^{i}_{t}S^{i}_{t}$.

For any process $P$, let
$\widetilde{P}:=\exp\left(-\int_{0}^{\cdot}r_{s}\ud s\right)P$ denote
its discounted incarnation. In terms of discounted quantities, the
wealth process has decomposition
\begin{equation}
\widetilde{X}=\widetilde{X}^{0}+F-C,  
\label{eq:wdecomp}
\end{equation}
where
\begin{equation}
\widetilde{X}^{0} := x + (H\cdot\widetilde{S})
\label{eq:sfw}  
\end{equation}
is the discounted wealth process of a self-financing portfolio
corresponding to strategy $H$, with
$(H\cdot\widetilde{S})\equiv\int_{0}^{\cdot}H_{s}\ud\widetilde{S}_{s}$
denoting the stochastic integral and
$C:=\int_{0}^{\cdot}\widetilde{c}_{s}\ud s,\,F
:=\int_{0}^{\cdot}\widetilde{f}_{s}\ud s$ denoting the non-decreasing
cumulative discounted consumption and income processes.

As is known from \cite{dv09}, because the income stream terminates
randomly, the agent is not able to follow the classical program of
borrowing against the present value of future income (so allowing
wealth to become negative) and using the optimal no-income strategy
with an initial wealth enlarged by the present value of future
income.
We shall therefore assume solvency at all
times, so $X\geq 0$ almost surely in \eqref{eq:wealth}. In this case,
for a given $x>0$, we call the pair $(H,c)$ (or $(X,c)$) an
\textit{$x$-admissible investment-consumption strategy}. Denote the
$x$-admissible investment-consumption strategies by $\mathcal{H}(x)$:
\begin{equation}
\mathcal{H}(x) := \left\{(H,c):\widetilde{X}:=x+(H\cdot\widetilde{S}) +
\int_{0}^{\cdot}(\widetilde{f}_{s} - \widetilde{c}_{s})\ud s\geq
0,\,\mbox{a.s}\right\}, \quad x>0. 
\label{eq:Hx}
\end{equation}
With an abuse of terminology we shall sometimes refer to the
wealth-consumption pair $(X,c)$ as an admissible
investment-consumption pair, and we shall sometimes write
$(X,c)\in\mathcal{H}(x)$ in place of $(H,c)\in\mathcal{H}(x)$. For
$x=1$, we write $\mathcal{H}\equiv\mathcal{H}(1)$.

If, for a consumption process $c$ we can find a predictable
$S$-integrable process $H$ such that $(H,c)\in\mathcal{H}(x)$ is an
$x$-admissible investment-consumption strategy, then we say that $c$
is an \textit{$x$-admissible consumption process} or, briefly, an
admissible consumption plan. Denote the set of $x$-admissible
consumption plans by $\mathcal{A}(x)$:
\begin{equation}
\mathcal{A}(x) := \left\{c\geq 0:\exists\, H \, \mbox{such
that} \, \widetilde{X}:=x+(H\cdot\widetilde{S}) +
\int_{0}^{\cdot}(\widetilde{f}_{s} - \widetilde{c}_{s})\ud s\geq
0,\,\mbox{a.s}\right\}, \quad x>0. 
\label{eq:Ax}
\end{equation}
For $x=1$ we write $\mathcal{A}\equiv\mathcal{A}(1)$.  It is easy to
verify that $\mathcal{A}(x)$ is a convex set.

For $c\equiv f\equiv 0$, the wealth process is that of a
self-financing portfolio, with discounted wealth process
$\widetilde{X}^{0}$ as in \eqref{eq:sfw}. Define $\mathcal{X}(x)$ as
the set of almost surely non-negative self-financing wealth processes
with initial value $x>0$:
\begin{equation*}
\mathcal{X}(x) := \left\{X^{0}: \widetilde{X}^{0} = x +
(H\cdot\widetilde{S})\geq 0,\,\mbox{a.s.}\right\}, \quad x>0. 
\end{equation*}
We write $\mathcal{X}\equiv\mathcal{X}(1)$, with
$\mathcal{X}(x)=x\mathcal{X}$ for $x>0$, and we note that
$\mathcal{X}$ is a convex set.

Given the wealth decomposition in \eqref{eq:wdecomp}, an equivalent
characterisation of the admissible consumption plans is that there
exists a self-financing wealth process such that its discounted
version plus cumulative discounted income dominates cumulative
discounted consumption.

\subsection{The primal problem}
\label{subsec:primal}

Let $U:\mathbb{R}_{+}\to\mathbb{R}$ be a utility function, strictly
concave, strictly increasing, continuously differentiable on
$\mathbb{R}_{+}$ and satisfying the Inada conditions
\begin{equation}
\lim_{x\downarrow 0}U^{\prime}(x) = +\infty, \quad
\lim_{x\to\infty}U^{\prime}(x) = 0.
\label{eq:inada}
\end{equation}

Let $\delta>0$ be an impatience factor for consumption. The
optimisation problem we study is to maximise expected utility from
consumption over the infinite horizon in the presence of the
terminating income stream. The primal value function
$u:\mathbb R_{+}\to\mathbb R$ is defined by
\begin{equation}
u(x) =  \sup_{c\in\mathcal{A}(x)}\mathbb{E}\left[\int_{0}^{\infty}
U(c_{t})\ud\kappa_{t}\right], \quad x>0,
\label{eq:pvf} 
\end{equation}
where $\kappa:\mathbb{R}_{+}\to\mathbb{R}_{+}$ is discounted Lebesgue
measure, given by
\begin{equation}
\kappa_{0}=0, \quad \ud\kappa_{t}=\e^{-\delta t}\ud t.
\label{eq:kappa}  
\end{equation}
For later use, define the positive process
$\zeta=(\zeta_{t})_{t\geq 0}$ as the reciprocal of
$(\ud\kappa_{t}/\ud t)_{t\geq 0}$:
\begin{equation}
\zeta_{t} := \left(\frac{\ud\kappa_{t}}{\ud t}\right)^{-1} =
\exp\left(\delta t\right), \quad t\geq 0.
\label{eq:alpha}
\end{equation}

Our goal is to develop a rigorous dual characterisation of the problem
in \eqref{eq:pvf}.

\subsection{Deflators and the budget constraint}
\label{subsec:dbc}

With the one-jump process $N$ in \eqref{eq:N} we associate the
non-negative c\`adl\`ag $(\mathbb{P},\mathbb{F})$-martingale $M$,
defined by
\begin{equation}
M_{t} := N_{t} + \eta\int_{0}^{t}N_{s}\ud s, \quad t\geq 0.
\label{eq:M}
\end{equation}

Let $Z^{\mathbb{G}}$ be a local martingale deflator for the asset
market in the absence of the income stream. Denote the set of such
deflators by $\mathcal{Z}^{\mathbb{G}}$. Each
$Z^{\mathbb{G}}\in \mathcal{Z}^{\mathbb{G}}$ is a positive local
martingale with unit initial value such that deflated discounted
self-financing wealth $Z^{\mathbb{G}} \widetilde{X}^{0}$ is a local
martingale, for each $X^{0}\in\mathcal{X}$.

Now incorporate the income stream, and let $\mathcal{Z}$ denote the
set of local martingale deflators for the enlarged market, such that
deflated discounted self-financing wealth is a local martingale. The
set $\mathcal{Z}$ is then composed of positive local martingales $Z$
given by
\begin{equation}
Z :=  Z^{\mathbb{G}}\mathcal{E}(-\gamma\cdot M),
\label{eq:Z}  
\end{equation}
where $\mathcal{E}(\cdot)$ denotes the stochastic exponential, for
c\`agl\`ad adapted processes $\gamma$ satisfying
$\gamma>-1,\,\gamma<+\infty$ almost surely. We note that as long as
$\gamma<+\infty$, $\Gamma:=\mathcal{E}(-\gamma\cdot M)$ is a
martingale.

The multiplicity of processes $Z\in\mathcal{Z}$ is the manifestation
of the market incompleteness induced both by the inherent incompleteness
of the asset market in the absence of the income, and also by the
presence of the randomly terminating income. The latter source of
incompleteness means that there is a multiplicity of integrands
$\gamma$ in \eqref{eq:Z}, as well as many processes
$Z^{\mathbb{G}}\in \mathcal{Z}^{\mathbb{G}}$.

For each $Z\in\mathcal{Z}$ and $X^{0}\in\mathcal{X}$, the process
$\widetilde{X}^{0}Z$ is a local martingale (and also a
super-martingale), so the (convex) set $\mathcal{Z}$ is defined by
\begin{equation}
\mathcal{Z} := \left\{Z>0,\,\mbox{c\`adl\`ag},\, Z_{0}=1:
\mbox{$\widetilde{X}^{0}Z$ is a local martingale, for all
$X^{0}\in\mathcal{X}$}\right\}. 
\label{eq:mcZ}
\end{equation}

For each $Z\in\mathcal{Z}$, we may define an associated
supermartingale deflator as the discounted martingale deflator
$\widetilde{Z}:=\exp\left(-\int_{0}^{\cdot}r_{s}\ud
  s\right)Z,\,Z\in\mathcal{Z}$, and we denote the set of such
supermartingales with initial value $y>0$ by
$\widetilde{\mathcal{Y}}(y)$:
\begin{equation}
\widetilde{\mathcal{Y}}(y) := \left\lbrace Y:
\mbox{$Y=y\widetilde{Z}=y\exp\left(-\int_{0}^{\cdot}r_{s}\ud
s\right)Z$, for $Z\in\mathcal{Z}$}\right\rbrace, \quad y>0. 
\label{eq:dlmd}
\end{equation}
We write $\widetilde{\mathcal{Y}}\equiv\widetilde{\mathcal{Y}}(1)$,
with $\widetilde{\mathcal{Y}}(y)=y\widetilde{\mathcal{Y}}$ for $y>0$,
and $\widetilde{\mathcal{Y}}$ is convex, inheriting this property from
$\mathcal{Z}$.

Define a further set $\mathcal{Y}^{0}(y)$ of supermartingale deflators
with initial value $y>0$ by
\begin{equation}
\mathcal{Y}^{0}(y) := \left\{Y>0,\,\mbox{c\`adl\`ag},\, Y_{0}=y:
\mbox{$X^{0}Y$ is a supermartingale, for all
$X^{0}\in\mathcal{X}$}\right\}, \quad y>0.
\label{eq:smd}
\end{equation}
As before, we write $\mathcal{Y}^{0}\equiv\mathcal{Y}^{0}(1)$, with
$\mathcal{Y}^{0}(y)=y\mathcal{Y}^{0}$ for $y>0$. Clearly, the set
$\mathcal{Y}^{0}$ is convex, and it includes all the processes in the
set $\widetilde{\mathcal{Y}}$ of discounted local martingale
deflators, but may include other processes. Since $X^{0}\equiv 1$ lies
in $\mathcal{X}$ (one may choose to hold initial wealth without
investing in stocks or the cash account), each $Y\in\mathcal{Y}^{0}$
is a supermartingale. We shall refer to $\mathcal{Y}^{0}$ as the set
of supermartingale deflators, or as the set of wealth deflators. The
processes $Y\in\mathcal{Y}^{0}$ correspond to the classical deflators
defined by \cite{ks99,ks03} in their seminal treatment of the terminal
wealth utility maximisation problem.  We thus have the inclusion
\begin{equation}
\mathcal{Y}^{0} \supseteq \widetilde{\mathcal{Y}}.
\label{eq:inclusion1}  
\end{equation}

The no-arbitrage assumption implicit in our model,
\begin{equation}
\mathcal{Y}^{0}(y) \neq \emptyset,
\label{eq:noarb}  
\end{equation}
is tantamount to the NUPBR
condition (see \cite{kk07}). We shall work only with deflators, and
will not require any arguments involving ELMMs. There are good reasons to do this. One is
that ELMMs will typically not exist over the infinite horizon, even if local martingale deflators are martingales, because these
martingales will  not be uniformly integrable when considered
over an infinite horizon (this is true even in the Black-Scholes
model, see  \cite[Section
1.7]{ks98}). It is now well accepted, since the work of
\cite{kk07} (and was implicit in the seminal work of \cite{klsx91},
where ELMMs were not used) that the key ingredient for well-posed
utility maximisation problems is the existence of a suitable class of
deflators which act on primal variables to create supermartingales,
and our approach is in this spirit. 

\subsubsection{Consumption deflators}
\label{subsubsec:consdefl}

We define the space which will form the dual domain for the
consumption problem \eqref{eq:pvf}.

For $x>0$, let $(X,c)\in\mathcal{H}(x)$ be an admissible
wealth-consumption pair, as defined in \eqref{eq:Hx}. Then
$c\in\mathcal{A}(x)$ is an admissible consumption process, as defined
in \eqref{eq:Ax}. The dual domain for the consumption problem
\eqref{eq:pvf} is defined by
\begin{equation}
\mathcal{Y}(y) := \left\{Y>0,\, \mbox{c\`adl\`ag},\, Y_{0}=y:
\mbox{$XY+\int_{0}(c_{s}-f_{s})Y_{s}\ud s$
is a supermartingale, $\forall\,c\in\mathcal{A}(x)$}\right\}.
\label{eq:mcY}  
\end{equation}
As usual, we write $\mathcal{Y}\equiv\mathcal{Y}(1)$ and we have
$\mathcal{Y}(y)=y\mathcal{Y}$ for $y>0$. The set $\mathcal{Y}$ is
easily seen to be convex. We shall refer to the processes in the dual
domain $\mathcal{Y}(y)$ as \textit{consumption deflators} (or, simply,
as deflators, when no confusion arises) to distinguish them from the
corresponding wealth deflators $Y\in\mathcal{Y}^{0}(y)$ as defined in
\eqref{eq:smd}, when the consumption and income processes are absent
(or are equal, so cancel out).

In \eqref{eq:mcY}, note that the wealth process $X$ is the one on the
left-hand-side of \eqref{eq:wealth}, so incorporating consumption and
income. Since $(X,c)\equiv(1,f)$ is an admissible
consumption-investment pair, each $Y\in\mathcal{Y}(y)$ is a
supermartingale. In particular, since $c\equiv f$ is an admissible
consumption plan, and noting the decomposition in \eqref{eq:wdecomp},
the resulting wealth process $X=X^{0}$ is then self-financing. In this
case we have that $XY=X^{0}Y$ is a supermartingale, so that
$Y\in\mathcal{Y}(y)$ is also a wealth deflator as defined in
\eqref{eq:smd}. In other words, we have the inclusion
\begin{equation}
\mathcal{Y} \subseteq \mathcal{Y}^{0}.
\label{eq:inclusion}
\end{equation}
The dual domain $\mathcal{Y}(y),\,y>0$ for our utility maximisation
problem \eqref{eq:pvf} from inter-temporal consumption in the presence
of the income stream is thus a specialisation of the one used by
\cite{ks99,ks03} for the terminal wealth problem.\footnote{As pointed
  out by an anonymous referee, it may even be the case that
  $\mathcal{Y}=\mathcal{Y}^{0}$. We have not been able to establish or
  refute this claim, which seems to be an interesting topic for a
  future research note.}

\subsubsection{The budget constraint}
\label{subsubsec:bc}

The key to identifying the dual problem to \eqref{eq:pvf} is a
suitable budget constraint involving admissible consumption plans and
some class of deflators. We show that such a budget constraint with
consumption deflators, due to the defining supermnartingale property
in \eqref{eq:mcY}. We later show that the same supermartingale
property also holds with discounted local martingale deflators
$Y\in\widetilde{\mathcal{Y}}$, so that the set of consumption
deflators includes the set of discounted local martingale deflators.

In what follows we assume that
$\mathbb{E}\left[\int_{0}^{\infty}f_{t}Y_{t}\ud
  t\right]<\infty,\,\forall\,Y\in\mathcal{Y}$, that is, cumulative
deflated income is integrable, and we note that this condition is
indeed true, as we establish it later in Lemma \ref{lem:Dbounded}.

\begin{lemma}[Budget constraint]
\label{lem:bc}

For $x,y>0$, let $c\in\mathcal{A}(x)$ and $Y\in\mathcal{Y}(y)$. We then
have the budget constraint
\begin{equation}
\mathbb{E}\left[\int_{0}^{\infty}(c_{t}-f_{t})Y_{t}\ud t\right]
\leq xy, \quad x,y>0.  
\label{eq:bc}
\end{equation}  

\end{lemma}

\begin{proof}

The supermartingale property in \eqref{eq:mcY} and the
non-negativity of $XY$ imply that 
\begin{equation*}
\mathbb{E}\left[\int_{0}^{t}c_{s}Y_{s}\ud s\right] \leq xy +
\mathbb{E}\left[\int_{0}^{t}f_{s}Y_{s}\ud s\right], \quad t\geq 0.
\end{equation*}
Letting $t\uparrow\infty$, using monotone convergence and
re-arranging, we obtain \eqref{eq:bc}.
  
\end{proof}

The set of consumption deflators includes the set of local martingale
deflators, since the supermartingale property in \eqref{eq:mcY} in
fact holds with discounted local martingale deflators, as we now show,
provided we assume that
$\mathbb{E}\left[\int_{0}^{\infty}f_{t}Y_{t}\ud
t\right]<\infty,\,\forall\,Y\in\widetilde{\mathcal{Y}}$.

\begin{lemma}[Budget constraint with discounted local martingale
deflators]
\label{lem:bcdlmd}

Let $x,y>0$. For any $c\in\mathcal{A}(x)$ and
$Y\in\widetilde{\mathcal{Y}}(y)$ we have the budget constraint
\begin{equation}
\mathbb{E}\left[\int_{0}^{\infty}(c_{t}-f_{t})Y_{t}\ud t\right]
\leq xy, \quad x,y>0.  
\label{eq:bcdlmd}
\end{equation}  

\end{lemma}

\begin{proof}

For some fixed $y>0$, let
$Y=y\widetilde{Z}\in\widetilde{\mathcal{Y}}(y)$ be a discounted local
martingale deflator, as defined in \eqref{eq:dlmd}. Recalling the
wealth dynamics \eqref{eq:wealth} and the discounted wealth
decomposition \eqref{eq:wdecomp}, the It\^o product rule applied to
$XY=y\widetilde{X}Z$ yields
\begin{equation}
XY + \int_{0}^{\cdot}(c_{s}-f_{s})Y_{s}\ud s = xy + X^{0}Y +
y\int_{0}^{\cdot}(F_{s-}-C_{s-})\ud Z_{s},
\label{eq:XY}
\end{equation}
where $X^{0}$ is the self-financing wealth process in the
decomposition \eqref{eq:wdecomp}. We observe that the right-hand-side
of \eqref{eq:XY} is a local martingale, since both
$X^{0}Y=y\widetilde{X}^{0}Z$ and the integral with respect to $Z$ are
local martingales. 

For any $t\geq 0$ the random variable $\int_{0}^{t}f_{s}Y_{s}\ud s$
is $\mathbb{P}$-integrable, since we have
\begin{equation*}
\mathbb{E}\left[\int_{0}^{t}f_{s}Y_{s}\ud s\right] \leq
\overline{a}\int_{0}^{t}\mathbb{E}[Y_{s}]\ud s <\infty, \quad t\geq 0.
\end{equation*}
Hence, the integrand on the left-hand-side of \eqref{eq:XY} is bounded
below by an integrable random variable. Since $XY$ is non-negative,
the right-hand side of \eqref{eq:XY} is a local martingale bounded
below by an integrable random variable, so the Fatou lemma gives that
\begin{equation}
\mbox{$XY + \int_{0}^{\cdot}(c_{s}-f_{s})Y_{s}\ud s$ is a
supermartingale, $\forall$ $c\in\mathcal{A}(x)$ and
$Y\in\widetilde{\mathcal{Y}}(y)$}.   
\label{eq:supermartingale}
\end{equation}
This supermartingale property then implies the budget constraint in
\eqref{eq:bcdlmd} by the same argument as in the proof of Lemma
\ref{lem:bc}.

\end{proof}

Since the supermartingale property in the defintion \eqref{eq:mcY}
also holds for discounted local martingale deflators (as in
\eqref{eq:supermartingale}), the set of consumption deflators contains
the set of discounted local martingale deflators. We thus have the
inclusion
\begin{equation}
\mathcal{Y} \supseteq \widetilde{\mathcal{Y}}.
\label{eq:inclusion2}
\end{equation}
We observe that, combining \eqref{eq:inclusion2} and
\eqref{eq:inclusion}, we have
$\mathcal{Y}^{0} \supseteq \mathcal{Y} \supseteq
\widetilde{\mathcal{Y}}$.

The budget constraint in \eqref{eq:bc} thus constitutes a necessary
condition for admissible consumption plans. We assume that it in fact
characterises admissible consumption plans, so is also a sufficient
condition, by virtue of the following \textit{financing condition}.

\begin{assumption}[Financing condition]
\label{ass:financing}

For $x,y>0$, let $(X,c)\in\mathcal{H}(x)$ be any $x$-admissible
wealth-consumption pair, so $c\in\mathcal{A}(x)$ is an admissible
consumption plan, and let $Y\in\mathcal{Y}(y)$ be any consumption
deflator. We assume that current wealth plus future income can finance
future consumption, in the sense that
\begin{equation*}
X_{t}Y_{t} + \mathbb{E}\left[\left.\int_{t}^{\infty}f_{s}Y_{s}\ud
s\right\vert\mathcal{F}_{t}\right] \geq
\mathbb{E}\left[\left.\int_{t}^{\infty}c_{s}Y_{s}\ud
s\right\vert\mathcal{F}_{t}\right], \quad t\geq 0,
\label{eq:financing}
\end{equation*}
and that if \eqref{eq:financing} holds, then $c\in\mathcal{A}(x)$ and
$(X,c)\in\mathcal{H}(x)$. 
  
\end{assumption}

\subsection{The dual problem}
\label{subsec:dual}

Let $V:\mathbb R_{+}\to\mathbb R$ denote the convex conjugate of the
utility function $U(\cdot)$, defined by
\begin{equation*}
V(y) := \sup_{x>0}[U(x)-xy], \quad y>0.  
\end{equation*}
The map $y\mapsto V(y),\,y>0$, is strictly convex, strictly
decreasing, continuously differentiable on $\mathbb{R}_{+}$,
$-V(\cdot)$ satisfies the Inada conditions, we have the bi-dual
relation $U(x) := \inf_{y>0}[V(y)+xy],\,x>0$, and
$V'(\cdot)=-I(\cdot)=-(U^{\prime})^{-1}(\cdot)$, where $I(\cdot)$
denotes the inverse of marginal utility. In particular, we have the
inequality
\begin{equation}
V(y) \geq U(x)-xy, \quad \forall\, x,y>0, \quad \mbox{with equality
iff $U^{\prime}(x)=y$}.  
\label{eq:VUbound}
\end{equation}

The dual to the primal problem \eqref{eq:pvf} is motivated in the
usual manner with the aid of the budget constraint \eqref{eq:bc}. For
any $c\in\mathcal{A}(x),\,Y\in\mathcal{Y}(y)$ we have, on recalling
the process $\zeta$ of \eqref{eq:alpha} and \eqref{eq:VUbound}, that
\begin{eqnarray}
\label{eq:vuineq}
\mathbb{E}\left[\int_{0}^{\infty}U(c_{t})\ud\kappa_{t}\right] & \leq &
 \mathbb{E}\left[\int_{0}^{\infty}\left(V(\zeta_{t}Y_{t}) +
f_{t}\zeta_{t}Y_{t}\right)\ud\kappa_{t}\right] + xy.
\end{eqnarray}
We therefore
define the dual value function by
\begin{equation}
v(y) :=
\inf_{Y\in\mathcal{Y}(y)}\mathbb{E}\left[\int_{0}^{\infty}\left(V(\zeta_{t}Y_{t})
+ f_{t}\zeta_{t}Y_{t}\right)\ud\kappa_{t}\right], \quad y>0. 
\label{eq:dvf}
\end{equation}
We shall assume throughout that the dual problem is finitely valued:
\begin{equation}
v(y) < \infty, \quad \forall \, y>0.
\label{eq:vfinite}
\end{equation}
It is well known that the condition \eqref{eq:vfinite} acts an
alternative mild feasibility condition to the reasonable asymptotic
elasticity condition of \cite{ks99} that ensures the usual tenets of a
duality theory can hold, as detailed by \cite{ks03}.

\begin{remark}[Consumption deflators versus discounted LMDs]
\label{rem:cdvslmd}

Since the budget constraint holds for both consumption deflators
$Y\in\mathcal{Y}$ and discounted local martingale deflators
$Y\in\widetilde{\mathcal{Y}}\subseteq\mathcal{Y}$, a natural question
to ask is whether the dual problem can be cast as a minimisation over
discounted local martingale deflators. It turns out that this is
possible provided one takes the \textit{closure} of
$\widetilde{\mathcal{Y}}$ with (with respect the topology of
convergence in measure $\kappa\times\mathbb{P}$) as the dual
domain. This is the content of Theorem \ref{thm:dense}, relying on the
property that $\widetilde{\mathcal{Y}}$ is dense in $\mathcal{Y}$ (see
Proposition \ref{prop:dense}).

The reason one must take the closure of $\widetilde{\mathcal{Y}}$ if
basing the dual problem on local martingale deflators is that it seems
hard to establish, in general, that $\widetilde{\mathcal{Y}}$ is a
closed set.  We shall prove that the dual domain $\mathcal{Y}$ is
closed by exploiting results on so-called Fatou convergence of
supermartingales. This is a primary reason for defining the dual
domain in terms of a fundamental supermartingale criterion. This
method of proof fails if using local martingale deflators, because the
 limiting supermartingale in the Fatou convergence
method is known only to be a supermartingale, and not necessarily a
local martingale deflator.

In our approach, we avoid the need to invoke a closure in defining the
dual domain, yielding a stronger ultimate duality statement, and in
some sense showing that we have identified the true space of dual
supermartingales. This is more in the spirit of \cite{ks99,ks03}. (As
\cite[\sc{Important Remark}, pp. 104--105]{rogers} points out, having
to invoke a closure in defining the dual domain weakens the statement
of the final result somewhat.) Finally, the fact that we show that our
dual domain actually coincides with the closure of the discounted
local martingale deflators, completes the picture on this topic in a
satisfying way.

\end{remark}

\subsection{The Davis-Vellekoop example}
\label{subsec:canonical}

A canonical example of the primal problem \eqref{eq:pvf} (and which
motivated this paper) is provided by \cite{dv09}. The underlying
market is a Black-Scholes (BS) market with a single stock $S$ with
constant market price of risk $\lambda\in\mathbb{R}$ and constant
volatility $\sigma>0$, driven by a single Brownian motion $W$. The
interest rate is a constant $r>0$, and the stock price dynamics are
given by
\begin{equation*}
\ud S_{t} = (r+\sigma\lambda)S_{t}\ud t + \sigma S_{t}\ud W_{t}.  
\end{equation*}

The income stream pays at a constant rate $a\geq 0$ until the random
time $\tau\sim\Exp(\eta)$, so $f_{t}=aN_{t},\,t\geq 0$, with $N$
defined in \eqref{eq:N} and $\tau$ independent of $W$. The underlying
market in the absence of the income is thus complete, so all the
incompleteness is generated by the random termination of the
income. The filtration $\mathbb{G}$ is the $\mathbb{P}$-augmentation of
the filtration generated by $W$, the unique martingale deflator in the
absence of the income is $Z^{\mathbb{G}}=\mathcal{E}(-\lambda W)$, so
the local martingale deflators in the market with the income are of
the form
\begin{equation*}
Z = \mathcal{E}(-\lambda W-\gamma\cdot M), \quad \gamma>-1,
\end{equation*}
with $M$ defined in \eqref{eq:M}. We take the filtration $\mathbb{F}$
to be the $\mathbb{P}$-augmentation of the filtration generated by
$(W,N)$. Over any finite horizon, $Z$ is a martingale as long as
$\gamma$ is finite. Over the infinite horizon it is well known that
such a martingale is not uniformly integrable. This is the case even
in the underlying BS market without the income stream, since the
martingale $\mathcal{E}(-\lambda W)$ is not uniformly integrable over
the infinite horizon, so ELMMs
will not exist over the infinite horizon. Thus, even in this simple
example, our approach of using deflators (as opposed to any
constructions involving ELMMs) is a very natural one.

This example is remarkable in that, despite the apparent simplicity of
the setting, the problem with randomly terminating income induces a
particularly awkward form of market incompleteness. The agent is
precluded from borrowing against the income stream, to avoid being
left insolvent if the income terminates at a time when the wealth is
negative. This raised many difficulties in \cite{dv09}. First, there
is no longer a closed form solution to the problem. Second,
\cite{dv09} encountered a duality gap, as their conjectured primal
solution, obtained by solving a deterministic control problem, could
not solve the HJB equation. Finally, there
was also an open question as to whether the marginal utility at zero
wealth is finite or not. Economic intuition indicates that it should
be finite as long as the income is strictly positive for a non-zero
initial time interval, but ODE heuristics in \cite{dv09} seemed to
suggest that the marginal utility at zero wealth was infinite, even though numerical solutions suggested the
opposite conclusion.

For all the above reasons, the problem that \cite{dv09} came up with
is a fascinating example, and we shall analyse it and give a numerical
solution in Section \ref{sec:dve}, after first establishing a rigorous
duality theory for the general semimartingale market model described
earlier. 


\section{The duality theorem}
\label{sec:tmd}

Here is the main duality result of the paper (Theorem
\ref{thm:cwrtid}), a dual characterisation of the solution to the
terminating income problem \eqref{eq:pvf}.
Note that our methods also cover the finite horizon version of the
problem, with or without a terminal wealth objective, and we give some
remarks on the adjustments needed to do this after proving the
theorem, in Section \ref{subsec:finiteh}.

\begin{theorem}[Consumption with randomly terminating income duality] 
\label{thm:cwrtid}

Define the primal value function $u(\cdot)$ by \eqref{eq:pvf} and the
dual value function $v(\cdot)$ by \eqref{eq:dvf}. Assume
\eqref{eq:inada}, \eqref{eq:noarb} and \eqref{eq:vfinite}. Define the
variable $y^{*}>0$ as the smallest value of $y>0$ at which the
derivative of the dual value function reaches zero:
\begin{equation}
y^{*} := \inf\{y>0: v^{\prime}(y)=0\}.  
\label{eq:y0}
\end{equation}

Then:

\begin{itemize}

\item[(i)] $u(\cdot)$ and $v(\cdot)$ are conjugate:
\begin{equation*}
v(y) = \sup_{x>0}[u(x)-xy], \quad u(x) =
\inf_{y\in(0,y^{*})}[v(y)+xy], \quad x>0,\, 0 < y < y^{*}.
\end{equation*}

\item[(ii)] The primal and dual optimisers
$\widehat{c}(x)\in\mathcal{A}(x)$ and
$\widehat{Y}(y)\in\mathcal{Y}(y)$ exist and are unique, so that
\begin{equation*}
u(x) =
\mathbb{E}\left[\int_{0}^{\infty}U(\widehat{c}_{t}(x))\ud\kappa_{t}\right],
\quad v(y) = \mathbb{E}\left[\int_{0}^{\infty}\left(
V(\zeta_{t}\widehat{Y}_{t}(y)) +
f_{t}\zeta_{t}\widehat{Y}_{t}(y)\right)\ud\kappa_{t}\right],
\quad x>0, \, y\in(0,y^{*}).
\end{equation*}

\item[(iii)] With $y=u^{\prime}(x)$ (equivalently,
$x=-v^{\prime}(y)$), the primal and dual optimisers are related
by
\begin{equation}
U^{\prime}(\widehat{c}_{t}(x)) = \zeta_{t}\widehat{Y}_{t}(y), \quad
\mbox{equivalently}, \quad \widehat{c}_{t}(x) =
-V^{\prime}(\zeta_{t}\widehat{Y}_{t}(y)), \quad t\geq 0,
\label{eq:pdc}
\end{equation}
and satisfy
\begin{equation}
\mathbb{E}\left[\int_{0}^{\infty}\left(\widehat{c}_{t}(x) -
f_{t}\right)\widehat{Y}_{t}(y)\ud t\right] = xy.
\label{eq:obc}
\end{equation}
Moreover, the associated optimal wealth process $\widehat{X}(x)$ is
given by
\begin{equation}
\widehat{X}_{t}(x)\widehat{Y}_{t}(y)  =
\mathbb{E}\left[\left.\int_{t}^{\infty}\left(\widehat{c}_{s}(x) -
f_{s}\right)\widehat{Y}_{s}(y)\ud s\right\vert\mathcal{F}_{t}\right],
\quad t\geq 0,
\label{eq:owp}
\end{equation}
and the process
$\widehat{X}(x)\widehat{Y}(y) +
\int_{0}^{\cdot}(\widehat{c}_{s}(x)-f_{s})\widehat{Y}_{s}(y)\ud s$ is
a uniformly integrable martingale, while
$\widehat{X}(x)\widehat{Y}(y)$ is a potential, a non-negative
supermartingale such that
$\lim_{t\to\infty}\mathbb{E}[\widehat{X}_{t}(x)\widehat{Y}_{t}(y)]=0$. Finally,
$\widehat{X}_{\infty}(x)\widehat{Y}_{\infty}(y)=0$ almost surely.

\item[(iv)] The functions $u(\cdot)$ and $-v(\cdot)$ are strictly
increasing, strictly concave and differentiable on their respective
domains, and the variable $y^{*}$ of \eqref{eq:y0} satisfies
$y^{*}<+\infty$, so that the primal Inada condition at zero is
violated:
\begin{equation*}
u^{\prime}(0) := \lim_{x\downarrow 0}u^{\prime}(x) < +\infty.
\end{equation*}
The primal Inada condition at infinity holds true, so that
\begin{equation*}
u^{\prime}(\infty) := \lim_{x\to\infty}u^{\prime}(x) = 0, \quad
-v^{\prime}(0) := \lim_{y\downarrow 0}(-v^{\prime}(y)) = +\infty.  
\end{equation*}
Moreover, the derivatives of the value functions satisfy
\begin{eqnarray*}
xu^{\prime}(x) & = & \mathbb{E}\left[\int_{0}^{\infty}
U^{\prime}(\widehat{c}_{t}(x))(\widehat{c}_{t}(x)-f_{t})\ud\kappa_{t}\right],
\quad x>0, \\
\quad yv^{\prime}(y) & = & \mathbb{E}\left[\int_{0}^{\infty}
\left(V^{\prime}(\zeta_{t}\widehat{Y}_{t}(y)) +
f_{t}\right)\widehat{Y}_{t}(y)\ud t\right], \quad y\in(0,y^{*}).
\end{eqnarray*}

\end{itemize}

\end{theorem}

The proof of Theorem \ref{thm:cwrtid} will be given in Section
\ref{sec:pdts}, and will proceed by proving an abstract version of the
theorem, which is stated in the next section.

We have used the set $\mathcal{Y}(y)$ of consumption deflators in the
definition \eqref{eq:dvf} of the dual value function, and we have the
inclusion \eqref{eq:inclusion2}.  The next theorem shows that the dual value function is an infimum over the closure of the set
$\widetilde{\mathcal{Y}}(y)$.  

\begin{theorem}
\label{thm:dense}

The dual value function \eqref{eq:dvf} also has the representation
\begin{equation*}
v(y) :=
\inf_{Y\in\cl(\widetilde{\mathcal{Y}}(y))}\mathbb{E}\left[\int_{0}^{\infty}
\left(V(\zeta_{t}Y_{t}) +
f_{t}\zeta_{t}Y_{t}\right)\ud\kappa_{t}\right], \quad y>0,  
\end{equation*}
where $\widetilde{\mathcal{Y}}(y)$ is the set of discounted local
martingale deflators defined in \eqref{eq:dlmd}, and $\cl(\cdot)$
denotes the closure with respect to convergence in measure
$\kappa\times\mathbb{P}$.

\end{theorem}

The proof of Theorem \ref{thm:dense} will be given in Section
\ref{sec:pdts} and will rest on Proposition \ref{prop:dense}, which
connects two dual domains in the abstract formulations of our
optimisation problems in Section \ref{sec:tad}.

\section{The abstract duality}
\label{sec:tad}

In this section we state an abstract duality theorem (Theorem
\ref{thm:adt}), from which Theorem \ref{thm:cwrtid} will
follow. Proofs of the results here will follow in Section
\ref{sec:pdts}.

Set $\mathbf{\Omega}:=[0,\infty)\times\Omega$. Let $\mathcal{G}$
denote the optional $\sigma$-algebra on $\mathbf{\Omega}$, that is,
the sub-$\sigma$-algebra of
$\mathcal{B}([0,\infty))\otimes\mathcal{F}$ generated by evanescent
sets and stochastic intervals of the form
$\llbracket T,\infty\llbracket$ for arbitrary stopping times
$T$. Define the measure $\mu:=\kappa\times\mathbb{P}$ on
$(\mathbf{\Omega},\mathcal{G})$. On the resulting finite measure space
$(\mathbf{\Omega},\mathcal{G},\mu)$, denote by $L^{0}_{+}(\mu)$ the
space of non-negative $\mu$-measurable functions, corresponding to
non-negative infinite horizon processes.

The primal and dual domains for our optimisation problems
\eqref{eq:pvf} and \eqref{eq:dvf} are now considered as subsets of
$L^{0}_{+}(\mu)$. The abstract primal domain $\mathcal{C}(x)$ is
identical to the set of admissible consumption plans, now considered
as a subset of $L^{0}_{+}(\mu)$:
\begin{equation}
\mathcal{C}(x) := \{g\in L^{0}_{+}(\mu): \mbox{$g=c,\,\mu$-a.e., for
some $c\in\mathcal{A}(x)$}\}, \quad x>0.  
\label{eq:Cx}
\end{equation}
As always we write $\mathcal{C}\equiv\mathcal{C}(1)$, and the set
$\mathcal{C}$ is convex. (Since $\mathcal{C}=\mathcal{A}$ we do not
really need to introduce the new notation, and do so only for some
notational symmetry in the abstract formulation.) In the abstract
notation, the primal value function \eqref{eq:pvf} is written as
\begin{equation}
u(x) := \sup_{g\in\mathcal{C}(x)}\int_{\mathbf{\Omega}}U(g)\ud\mu,
\quad x>0.
\label{eq:vfabs}
\end{equation}

For the dual problem, the abstract dual domain $\mathcal{D}(y),\,y>0$
is composed of elements $h=\zeta Y$ appearing in the dual value
function \eqref{eq:dvf}. We thus define
\begin{equation}
\mathcal{D}(y) := \{h\in L^{0}_{+}(\mu): \mbox{$h=\zeta
Y,\,\mu$-a.e., for some $Y\in\mathcal{Y}(y)$}\}, \quad y>0.
\label{eq:Dy}
\end{equation}
As usual, we write $\mathcal{D}\equiv\mathcal{D}(1)$, we have
$\mathcal{D}(y)=y\mathcal{D}$ for $y>0$, and the set $\mathcal{D}$ is
convex. With this notation, the dual problem \eqref{eq:dvf} takes the
form
\begin{equation}
v(y) := \inf_{h\in\mathcal{D}(y)}\int_{\mathbf{\Omega}}\left(V(h) +
fh\right)\ud\mu, \quad y>0. 
\label{eq:dvfabs}
\end{equation}

\begin{remark}
\label{rem:solid}
  
Observe that we have not enlarged the original primal or (in
particular) the dual domain in the classical manner akin to
\cite{ks99,ks03}, to encompass processes dominated by elements of the
original domain. (So, for example, on the dual side, one might define
$\mathcal{D}$ to be composed of all elements $h$ such that
$h\leq\zeta Y$, for some $Y\in\mathcal{Y}$.) Such enlargements are not
operative here, for a number of reasons, all stemming from the
presence of the random endowment $f$, as we now describe.

First, our duality proof will not be based on an
application of the bipolar theorem of \cite{bs99}. One reason for this
is that the primal variable $c-f$ in the budget constraint
\eqref{eq:bc}, is not necessarily non-negative, which precludes the
use of the bipolar theorem, as that applies to elements in
$L^{0}_{+}(\mu)$. In this context, the enlargement of the dual
domain, to encompass processes dominated by the original dual
variables, is typically carried out in order to reach the bipolar of
the original dual domain, so as to use the bipolar
theorem, which is of no use to us.

Second, in problems without random endowment, the dual enlargement,
which renders the abstract dual domain solid,\footnote{Recall that a
set $A\subseteq L^{0}_{+}(\mu)$ is called \emph{solid} if
$\overline{h}\in A$ and $0\leq h\leq
\overline{h},\,\mu$-a.e. implies that $h\in A$.} does not prevent
the dual optimiser from lying in the original dual domain (in essence,
the monotone decreasing property of $V(\cdot)$ ensures that one takes
the ``largest'' dual element in order to find the dual optimiser,
which then lies in the original dual domain). But in a problem with
random endowment, as in \eqref{eq:dvfabs}, the dual objective function
is no longer monotone decreasing in $h\in\mathcal{D}(y)$, so the
optimiser in $\mathcal{D}(y)$ might not lie in the original dual
domain if $\mathcal{D}(y)$ contains all elements $h\leq\zeta Y$, for
some $Y\in\mathcal{Y}(y)$.

These features illustrate some of the difficulties that arise in
utility maximisation problems with random endowment, which have made
such problems hard to deal with, and go some way to explaining the
subtle techniques that have had to be employed to obtain results in
this area. These include the use of finitely additive measures, as in
\cite{csw01,csw17,kz03}, or the expansion of the dimension of the
value function with an additional variable for the number of units of
the random endowment, as in \cite{hk04,most17,mostsirbu20}. One of the
contributions of this paper is to obtain very general duality results
without having to introduce such remedies.
  
\end{remark}

For later use, and in particular to state a result further below
(Proposition \ref{prop:dense}) that will ultimately furnish us with
the proof of Theorem \ref{thm:dense}, we define the abstract domain
$\widetilde{\mathcal{D}}(y),\,y>0$ in an analogous manner to
\eqref{eq:Dy}, as the counterpart to the set
$\mathcal{\widetilde{Y}}(y)$ of discounted martingale deflators:
\begin{equation}
\widetilde{\mathcal{D}}(y) := \{h\in L^{0}_{+}(\mu): \mbox{$h=\zeta
Y,\,\mu$-a.e., for some $Y\in\widetilde{\mathcal{Y}}(y)$}\}, \quad y>0.
\label{eq:mcDy}
\end{equation}
As usual, we write
$\widetilde{\mathcal{D}}\equiv\widetilde{\mathcal{D}}(1)$, we have
$\widetilde{\mathcal{D}}(y)=y\widetilde{\mathcal{D}}$ for $y>0$, and
the set $\widetilde{\mathcal{D}}$ is convex.

The abstract duality theorem will rely on certain basic properties of
the sets $\mathcal{C}$ and $\mathcal{D}$, which we state in the key
Proposition \ref{prop:propertiesCD} below. In what follows we shall
sometimes employ the notation
$\langle g,h\rangle:=\int_{\mathbf{\Omega}}gh\ud\mu,\,g,h\in
L^{0}(\mu)$, and we shall call a set $A\subseteq L^{0}_{+}(\mu)$
\emph{closed in $\mu$-measure}, or simply \emph{closed}, if it is
closed with respect to the topology of convergence in measure $\mu$.

\begin{proposition}[Properties of the primal and dual domains]
\label{prop:propertiesCD}
  
The primal domain $\mathcal{C}$ satisfies
\begin{equation}
g \in \mathcal{C} \iff \langle g-f,h\rangle \leq 1, \quad \forall\,
h\in\mathcal{D}.
\label{eq:dualC}
\end{equation}
Both $\mathcal{C}$ and the abstract dual domain $\mathcal{D}$ are
closed, convex subsets of $L^{0}_{+}(\mu)$, and are bounded in
$L^{0}(\mu)$. Moreover, $\mathcal{D}$ is bounded in $L^{1}(\mu)$.

\end{proposition}

The proof of Proposition \ref{prop:propertiesCD} will be given in
Section \ref{sec:pdts}, using Lemmata \ref{lem:Dclosed} and
\ref{lem:Dbounded}.

\begin{theorem}[Abstract duality theorem]
\label{thm:adt}

Define the primal value function $u(\cdot)$ by \eqref{eq:vfabs} and
the dual value function $v(\cdot)$ by \eqref{eq:dvfabs}. Assume the
Inada conditions \eqref{eq:inada} and that
\begin{equation*}
u(x) > -\infty,\,\forall \, x>0, \quad v(y)< \infty,\,\forall \,y>0.   
\end{equation*}

Set
\begin{equation}
y^{*} := \inf\{y>0: v^{\prime}(y)=0\}.  
\label{eq:y0abs}
\end{equation}
Then, with Proposition \ref{prop:propertiesCD} in place, we have:

\begin{itemize}

\item[(i)] $u(\cdot)$ and $v(\cdot)$ are conjugate:
\begin{equation}
v(y) = \sup_{x>0}[u(x)-xy], \quad u(x) = \inf_{y\in(0,y^{*})}[v(y)+xy], \quad
x>0,\,y\in(0,y^{*}).
\label{eq:conjugacy}
\end{equation}

\item[(ii)] The primal and dual optimisers
$\widehat{g}(x)\in\mathcal{C}(x)$ and
$\widehat{h}(y)\in\mathcal{D}(y)$ exist and are unique, so that
\begin{equation*}
u(x) = \int_{\mathbf{\Omega}}U(\widehat{g}(x))\ud\mu, \quad v(y) =
\int_{\mathbf{\Omega}}\left(V(\widehat{h}(y)) +
f\widehat{h}(y)\right)\ud\mu, \quad x>0,\,y\in(0,y^{*}). 
\end{equation*}

\item[(iii)] With $y=u^{\prime}(x)$ (equivalently,
$x=-v^{\prime}(y)$), the primal and dual optimisers are related by
\begin{equation*}
U^{\prime}(\widehat{g}(x)) = \widehat{h}(y), \quad
\mbox{equivalently}, \quad \widehat{g}(x) =
-V^{\prime}(\widehat{h}(y)), 
\end{equation*}
and satisfy
\begin{equation*}
\langle \widehat{g}(x)-f,\widehat{h}(y)\rangle = xy.
\end{equation*}

\item[(iv)] $u(\cdot)$ and $-v(\cdot)$ are strictly increasing,
strictly concave and differentiable on their respective domains. The
constant $y^{*}$ in \eqref{eq:y0abs} is finite: $y^{*}<+\infty$, so
that the primal value function has finite derivative at zero:
\begin{equation*}
u^{\prime}(0) := \lim_{x\downarrow 0}u^{\prime}(x) < +\infty,  
\end{equation*}
while the derivative of the primal value function at infinity and of
the dual value function at zero satisfy
\begin{equation*}
u^{\prime}(\infty) := \lim_{x\to\infty}u^{\prime}(x) = 0, \quad
-v^{\prime}(0) := \lim_{y\downarrow 0}(-v^{\prime}(y)) = +\infty.
\end{equation*}
Furthermore, the derivatives of the value functions satisfy
\begin{equation*}
xu^{\prime}(x) = \int_{\mathbf{\Omega}}
U^{\prime}(\widehat{g}(x))(\widehat{g}(x)-f)\ud\mu, \quad
yv^{\prime}(y) = \int_{\mathbf{\Omega}}
\left(V^{\prime}(\widehat{h}(y)) + f\right)\widehat{h}(y)\ud\mu, \quad
x>0,\,y\in(0,y^{*}).  
\end{equation*}

\end{itemize}

\end{theorem}

The proof of Theorem \ref{thm:adt} will follow in Section
\ref{sec:pdts}.

Theorem \ref{thm:dense} will rest on the following proposition, which
connects the sets $\mathcal{D}$ and $\widetilde{\mathcal{D}}$.

\begin{proposition}
\label{prop:dense}

With respect to the topology of convergence in measure $\mu$, the set
$\widetilde{\mathcal{D}}\equiv\widetilde{\mathcal{D}}(1)$ of
\eqref{eq:mcDy} is dense in the abstract dual domain
$\mathcal{D}\equiv\mathcal{D}(1)$ of \eqref{eq:Dy}. 
That is, we have
\begin{equation*}
\mathcal{D} = \cl(\widetilde{\mathcal{D}}),
\end{equation*}
where $\cl(\cdot)$ denotes the closure with respect to the topology of
convergence in measure $\mu$.

\end{proposition}

The proof of Proposition \ref{prop:dense} will be given in Section
\ref{sec:pdts}, following the proofs of the abstract and concrete
duality theorems.

\section{Proofs of the duality theorems}
\label{sec:pdts}

In this section, we prove the abstract duality of Theorem
\ref{thm:adt}, and then establish the concrete duality of Theorem
\ref{thm:cwrtid}. We proceed via a series of lemmas. Using the fact
that the budget constraint \eqref{eq:bc} is a sufficient as well as
necessary condition for admissibility, along with two subsequent
results (Lemma \ref{lem:Dclosed} and Lemma \ref{lem:Dbounded}), that
the dual domain $\mathcal{D}$ is closed, and also bounded in
$L^{1}(\mu)$, we establish Proposition \ref{prop:propertiesCD}, which
underpins the abstract duality.

\subsection{Properties of the dual domain}
\label{subsec:Dclosed}

\subsubsection{Closed property of the dual domain}
\label{subsubsec:Dclosed}

The first step towards establishing Proposition
\ref{prop:propertiesCD} is to show that $\mathcal{D}$ is closed. As in
some classical proofs \cite[Lemma 4.1]{ks99}, we shall employ
supermartingale convergence results based on Fatou convergence of
processes.

Here is the closed property for the abstract dual domain.

\begin{lemma}
\label{lem:Dclosed}

The dual domain $\mathcal{D}\equiv\mathcal{D}(1)$ of \eqref{eq:Dy} is
closed with respect to the topology of convergence in measure $\mu$.
  
\end{lemma}

\begin{proof}

Let $(h^{n})_{n\in\mathbb{N}}$ be a sequence in $\mathcal{D}$,
converging $\mu$-a.e. to some $h\in L^{0}_{+}(\mu)$. We want to show
that $h\in\mathcal{D}$.

Since $h^{n}\in\mathcal{D}$, for each $n\in\mathbb{N}$ we have
$h^{n}=\zeta\widehat{Y}^{n},\,\mu$-a.e for some supermartingale
$\widehat{Y}^{n}\in\mathcal{Y}$. With $\mathbb{T}$ a dense countable
subset of $\mathbb{R}_{+}$, \cite[Lemma 5.2]{fk97} implies that there
exists a sequence $(Y^{n})_{n\in\mathbb{N}}$ of supermartingales, with
each $Y^{n}\in\conv(\widehat{Y}^{n},\widehat{Y}^{n+1},\ldots)$, where
$\conv(\widehat{Y}^{n},\widehat{Y}^{n+1},\ldots)$ denotes a convex
combination $\sum_{k=n}^{N(n)}\lambda_{k}\widehat{Y}^{k}$ for
$\lambda_{k}\in[0,1]$ with $\sum_{k=n}^{N(n)}\lambda_{k}=1$, and a
supermartingale $Y$, such that $(Y^{n})_{n\in\mathbb{N}}$ is Fatou
convergent on $\mathbb{T}$ to $Y$.

Define a supermartingale sequence $(\widehat{V}^{n})_{n\in\mathbb{N}}$
by
$\widehat{V}^{n}:=X\widehat{Y}^{n} +
\int_{0}^{\cdot}(c_{s}-f_{s})\widehat{Y}^{n}_{s}\ud s$, with
$c\in\mathcal{A}$ an admissible consumption plan and $X$ the
associated wealth process, so $(X,c)\in\mathcal{H}$ is an admissible
investment-consumption pair with initial wealth $1$. Once again from
\cite[Lemma 5.2]{fk97} there exists a sequence
$(V^{n})_{n\in\mathbb{N}}$ of supermartingales with each
$V^{n}\in\conv(\widehat{V}^{n},\widehat{V}^{n+1},\ldots)$, and a
supermartingale $V$, such that $(V^{n})_{n\in\mathbb{N}}$ is Fatou
convergent on $\mathbb{T}$ to $V$. We observe that, because
$Y^{n}\in\conv(\widehat{Y}^{n},\widehat{Y}^{n+1},\ldots)$ and
$V^{n}\in\conv(\widehat{V}^{n},\widehat{V}^{n+1},\ldots)$, we have
\begin{equation*}
V^{n} = \sum_{k=n}^{N(n)}\lambda_{k}\widehat{V}^{k} =
\sum_{k=n}^{N(n)}\lambda_{k}
\left(X\widehat{Y}^{k}+\int_{0}^{\cdot}(c_{s}-f_{s})\widehat{Y}^{k}_{s}\ud
  s\right) = XY^{n} + \int_{0}^{\cdot}c_{s}Y^{n}_{s}\ud s.
\end{equation*}
As a consequence, and because the sequence $(Y^{n})_{n\in\mathbb{N}}$
is Fatou convergent on $\mathbb{T}$ to the supermartingale $Y$, the
sequence
$(V^{n})_{n\in\mathbb{N}}=\left(XY^{n}+\int_{0}^{\cdot}(c_{s}-f_{s})Y^{n}_{s}\ud
  s\right)_{n\in\mathbb{N}}$ is Fatou convergent on $\mathbb{T}$ to
the supermartingale $V=XY+\int_{0}^{\cdot}(c_{s}-f_{s})Y_{s}\ud s$. Since
$XY+\int_{0}^{\cdot}(c_{s}-f_{s})Y_{s}\ud s$ is a supermartingale and
$c\in\mathcal{A}$ is an admissible consumption plan (so
$(X,c)$ is an admissible investment-consumption strategy), we have
$Y\in\mathcal{Y}$.

Because $h^{n}=\zeta\widehat{Y}^{n},\,\mu$-a.e.for each
$n\in\mathbb{N}$, and since
$Y^{n}=\sum_{k=n}^{N(n)}\lambda_{k}\widehat{Y}^{k}$, we have
\begin{equation}
\zeta Y^{n} =
\sum_{k=n}^{N(n)}\lambda_{k}\zeta\widehat{Y}^{k}
= \sum_{k=n}^{N(n)}\lambda_{k}h^{k}, \quad \mbox{$\mu$-a.e.} 
\label{eq:conv}
\end{equation}
Now, by \cite[Lemma 8]{z02} (proven there for finite
horizon processes, but it is straightforward to verify that the proof
goes through without alteration for infinite horizon processes) there
is a countable set $K\subset\mathbb{R}_{+}$ such that for
$t\in\mathbb{R}_{+}\setminus K$, we have
$Y_{t}=\liminf_{n\to\infty}Y^{n}_{t}$ almost surely, and hence also
$Y=\liminf_{n\to\infty}Y^{n}$, $\mu$-almost everywhere (since these
differ only on a set of measure zero). So, taking the limit inferior
in \eqref{eq:conv} and recalling that $(h^{n})_{n\in\mathbb{N}}$
converges $\mu$-a.e. to $h$, we obtain
\begin{equation*}
\zeta Y = \liminf_{n\to\infty}\zeta Y^{n} =
\liminf_{n\to\infty}\sum_{k=n}^{N(n)}\lambda_{k}h^{k} = h, \quad
\mbox{$\mu$-a.e.}    
\label{eq:conv1}
\end{equation*}
That is, $h=\zeta Y$ $\mu$-a.e, for $Y\in\mathcal{Y}$, so
$h\in\mathcal{D}$, and thus $\mathcal{D}$ is closed.

\end{proof}

\subsubsection{$L^{1}(\mu)$-boundedness of the dual domain}
\label{subsubsec:Dbdd}


The remaining ingredient we shall need to prove
\ref{prop:propertiesCD} is to show that the dual domain is bounded in
$L^{1}(\mu)$. This is taken care of by the lemma below, on the
boundedness properties of two integrals involving elements
$h\in\mathcal{D}$. As these results will be used at various places in
the duality proof, we collect them here.

\begin{lemma}[Bounded integrals in the dual domain]
\label{lem:Dbounded}

Recall the lower bound in \eqref{eq:rbound} for the interest rate and
the upper bound in \eqref{eq:fbound} on the income rate. We then have
\begin{equation}
\int_{\mathbf{\Omega}}h\ud\mu \leq \frac{1}{\underline{r}}, \quad
\int_{\mathbf{\Omega}}fh\ud\mu \leq
\frac{\overline{a}}{\underline{r}}, \quad \forall \, h\in\mathcal{D}. 
\label{eq:Dbounds}
\end{equation}
  
\end{lemma}

\begin{proof}

Choose a consumption plan such that $c_{t}-f_{t}=\epsilon,\,t\geq
0$, for some constant $\epsilon\in(0,\underline{r}]$. It is
straightforward to verify that this plan is admissible, as
follows. With initial wealth $x=1$ and trading strategy $H\equiv 0$,
the lower bound in \eqref{eq:rbound} on the interest rate yields that
\begin{equation*}
X_{t} \geq \frac{\epsilon}{\underline{r}} +
\exp(\underline{r}t)\left(1 -  \frac{\epsilon}{\underline{r}}\right),
\quad t\geq 0.   
\end{equation*}
Thus, provided $\epsilon\in(0,\underline{r}]$, we have
$X_{t}\geq 0,\,\forall\,t\geq 0$ almost surely, so that the chosen
consumption plan is admissible: $c\in\mathcal{A}$. In particular, we
may set $\epsilon=\underline{r}$, which yields
$X_{t}\geq 1,\,\forall\,t\geq 0$ almost surely. (This makes perfect sense:
one can consume at a rate equal to the income rate plus the minimal
interest rate and stay solvent.)

The budget constraint \eqref{eq:bc} implies that
\begin{equation}
\int_{\mathbf{\Omega}}(g-f)h\ud\mu \leq 1, \quad \forall \,
g\in\mathcal{C},\, h\in\mathcal{D}.
\label{eq:bcabs}
\end{equation}
So, if we choose the consumption plan such that $g-f=\underline{r}$,
$\mu$-a.e., we have $g\in\mathcal{C}$ from the first part of the
proof, and the abstract budget constraint in \eqref{eq:bcabs} converts
to
\begin{equation*}
\underline{r}\int_{\mathbf{\Omega}}h\ud\mu \leq 1, \quad \forall\,
h\in\mathcal{D}, 
\end{equation*}
so we obtain the first inequality in \eqref{eq:Dbounds}. The second
inequality then follows by augmenting this with the bound in
\eqref{eq:fbound} on the income rate.

\end{proof}

  


With this preparation, we are now able to establish Proposition
\ref{prop:propertiesCD}. 

\begin{proof}[Proof of Proposition \ref{prop:propertiesCD}]

We first establish the dual characterisation of $\mathcal{C}$ as
expressed in \eqref{eq:dualC}. The budget constraint \eqref{eq:bc}
along with the financing condition in Assumption \ref{ass:financing},
give the equivalence, invoking the measure $\kappa$ of
\eqref{eq:kappa},
\begin{equation*}
c\in\mathcal{A} \iff
\mathbb{E}\left[\int_{0}^{T}(c_{t}-f_{t})\zeta_{t}Y_{t}\ud\kappa_{t}\right]
\leq 1, \quad \forall\,Y\in\mathcal{Y}.  
\end{equation*}
So, with $\mathcal{C}\owns g\equiv c$ and $\mathcal{D}\owns\zeta Y$,
in terms of the measure $\mu$ we have
\begin{equation}
g\in\mathcal{C} \iff 
\int_{\mathbf{\Omega}}(g-f)h\ud\mu \leq 1, \quad \forall\,
h\in \mathcal{D},
\label{eq:Cchar0}
\end{equation}
which establishes \eqref{eq:dualC}.

The equivalence \eqref{eq:Cchar0} along with Fatou's lemma yields that
the set $\mathcal{C}$ is closed with respect to the topology of
convergence in measure $\mu$. To see this, let
$(g^{n})_{n\in\mathbb{N}}$ be a sequence in $\mathcal{C}$ which
converges $\mu$-a.e. to an element $g\in L^{0}_{+}(\mu)$. For
arbitrary $h\in\mathcal{D}$ we obtain, via Fatou's lemma and the fact
that $g^{n}\in\mathcal{C}$ for each $n\in\mathbb{N}$ (with $g^{n}-f$
bounded below),
\begin{equation*}
\int_{\mathbf{\Omega}}(g-f)h\ud\mu \leq
\liminf_{n\to\infty}\int_{\mathbf{\Omega}}(g^{n}-f)h\ud\mu \leq 1,
\end{equation*}
so by \eqref{eq:Cchar0}, $g\in\mathcal{C}$, and thus $\mathcal{C}$ is
closed. Convexity of $\mathcal{C}$ is clear from its definition.
  
For the $L^{0}(\mu)$-boundedness of $\mathcal{C}$, we shall find a
positive element $\overline{h}\in\mathcal{D}$ and show that
$\mathcal{C}$ is bounded in $L^{1}(\overline{h}\ud\mu)$ and hence
bounded in $L^{0}(\mu)$. Given the inclusion in \eqref{eq:inclusion2},
we may choose $\mathcal{Z}\owns Z\equiv 1$, so that
$Y=\exp\left(-\int_{0}^{\cdot}r_{s}\ud
  s\right)\in\widetilde{\mathcal{Y}}\subseteq\mathcal{Y}$, and then
$\overline{h}:=\zeta\exp\left(-\int_{0}^{\cdot}r_{s}\ud
  s\right)\in\mathcal{D}$ defines a strictly positive element of
$\mathcal{D}$. Observe that, with the lower bound on the interest rate
in \eqref{eq:rbound}
\begin{equation*}
\int_{\mathbf{\Omega}}\overline{h}\ud\mu =
\mathbb{E}\left[\int_{0}^{\infty}\exp\left(-\int_{0}^{t}r_{s}\ud s\right)
\zeta_{t}\ud\kappa_{t}\right] \leq
\mathbb{E}\left[\int_{0}^{\infty}\exp\left(-\underline{r}t\right)\ud
t\right] = \frac{1}{\underline{r}}. 
\end{equation*}

By virtue of the budget constraint we have, for any $g\in\mathcal{C}$,
that $\int_{\mathbf{\Omega}}(g-f)\overline{h}\ud\mu\leq 1$.Thus, with
the bound on the income rate in \eqref{eq:fbound}, we have
\begin{equation}
\int_{\mathbf{\Omega}}g\overline{h}\ud\mu \leq 1 +
\int_{\mathbf{\Omega}}f\overline{h}\ud\mu \leq 1 + \overline{a}
\int_{\mathbf{\Omega}}\overline{h}\ud\mu = 1 +
\frac{\overline{a}}{\underline{r}} <\infty.   
\label{eq:boundghbar}
\end{equation}
Thus, $\mathcal{C}$ is bounded in $L^{1}(\overline{h}\ud\mu)$ and
hence bounded in $L^{0}(\mu)$.

Finally, we note from Lemma \ref{lem:Dbounded}, and in particular the
first integral in \eqref{eq:Dbounds}, that $\mathcal{D}$ is bounded in
$L^{1}(\mu)$, and hence also bounded in $L^{0}(\mu)$.

\end{proof}

\subsection{The abstract duality proof}
\label{subsec:adp}

We now take Proposition \ref{prop:propertiesCD} as given in the
remainder of this section, and proceed with the proof of the abstract
duality of Theorem \ref{thm:adt} via a series of lemmas.

The first step is to establish weak duality.

\begin{lemma}[Weak duality]
\label{lem:weakdual}
  
The primal and dual value functions $u(\cdot)$ and $v(\cdot)$ of
\eqref{eq:vfabs} and \eqref{eq:dvfabs} satisfy the weak duality bounds
\begin{equation}
v(y) \geq \sup_{x>0}[u(x)-xy], \quad y>0, \quad
\mbox{equivalently} \quad u(x) \leq \inf_{y>0}[v(y)+xy],
\quad x>0. 
\label{eq:weakd}
\end{equation}
As a result, $u(x)$ is finitely valued for all $x>0$. Moreover, we
have the limiting relations
\begin{equation}
\limsup_{x\to\infty}\frac{u(x)}{x} \leq 0, \quad
\liminf_{y\to\infty}\frac{v(y)}{y} \geq 0. 
\label{eq:limiting}
\end{equation}

\end{lemma}

\begin{proof}

For any $g\in\mathcal{C}(x)$ and $h\in\mathcal{D}(y)$, using the
budget constraint $\int_{\mathbf{\Omega}}(g-f)h\ud\mu\leq xy$ in the
same manner as the arguments leading to \eqref{eq:vuineq}, we may
bound the achievable utility according to
\begin{equation}
\int_{\mathbf{\Omega}}U(g)\ud\mu \leq \int_{\mathbf{\Omega}}\left(V(h)
+ fh\right)\ud\mu + xy, \quad x,y>0,
\label{eq:ucbound}
\end{equation}
Maximising the left-hand-side of \eqref{eq:ucbound} over
$g\in\mathcal{C}(x)$ and minimising the right-hand-side over
$h\in\mathcal{D}(y)$ gives $u(x)\leq v(y)+xy$, and \eqref{eq:weakd}
follows.

The assumption that $v(y)<\infty$ for all $y>0$ immediately yields
that $u(x)$ is finitely valued for some $x>0$. Since $U(\cdot)$ is
strictly increasing and strictly concave, and given the convexity of
$\mathcal{C}$, these properties are inherited by $u(\cdot)$, which is
therefore finitely valued for all $x>0$. Finally, the relations in
\eqref{eq:weakd} easily lead to those in \eqref{eq:limiting}.
  
\end{proof}

The next step is to give a compactness lemma for the primal
domain. The proof is on similar lines to the proof (for the dual domain) of
\cite[Lemma 3.6]{most15}, and uses \cite[Lemma A1.1]{ds94} (adapted
from a probability space to the finite measure space
$(\mathbf{\Omega},\mathcal{G},\mu)$), so for brevity is omitted.

\begin{lemma}[Compactness lemma for $\mathcal{C}$]
\label{lem:Ccompact}

Let $(\tilde{g}^{n})_{n\in\mathbb{N}}$ be a sequence in
$\mathcal{C}$. Then there exists a sequence $(g^{n})_{n\in\mathbb{N}}$
with $g^{n}\in\conv(\tilde{g}^{n}, \tilde{g}^{n+1},\ldots)$, which
converges $\mu$-a.e. to an element $g\in\mathcal{C}$ that is
$\mu$-a.e. finite.

\end{lemma}






Here is the next step in this chain of results, a uniform
integrability property associated with elements of $\mathcal{C}$ and
the positive part of the utility function.

\begin{lemma}[Uniform integrability of
$(U^{+}(g^{n}))_{n\in\mathbb{N}},\,g^{n}\in\mathcal{C}(x)$]
\label{lem:Uplusg}

The family $(U^{+}(g))_{g\in\mathcal{C}(x)}$ is uniformly integrable,
for any $x>0$.
  
\end{lemma}

\begin{proof}

Fix $x>0$. If $U(\infty)\leq 0$ there is nothing to prove, so assume
$U(\infty)>0$.

If the sequence $(U^{+}(g^{n}))_{n\in\mathbb{N}}$ is not uniformly
integrable, then, passing if need be to a subsequence still denoted by 
$(g^{n})_{n\in\mathbb{N}}$, we can find a constant $\alpha>0$ and a
disjoint sequence $(A_{n})_{n\in\mathbb{N}}$ of sets of
$(\mathbf{\Omega},\mathcal{G})$ (so
$A_{n}\in\mathcal{G},\,n\in\mathbb{N}$ and $A_{i}\cap A_{j}=\emptyset$
if $i\neq j$) such that
\begin{equation*}
\int_{\mathbf{\Omega}}U^{+}(g^{n})\mathbbm{1}_{A_{n}}\ud\mu \geq
\alpha, \quad n\in\mathbb{N}.  
\end{equation*}
(See for example  \cite[Corollary A.1.1]{pham09}.) Define a
sequence $(f^{n})_{n\in\mathbb{N}}\in L^{0}_{+}(\mu)$ by
\begin{equation*}
f^{n} := x_{0} + \sum_{k=1}^{n}g^{k}\mathbbm{1}_{A_{k}}, \quad
x_{0}:=\inf\{x>0:\,U(x)\geq 0\}. 
\end{equation*}
For any $h\in\mathcal{D}$ (so
$\int_{\mathbf{\Omega}}h\ud\mu\leq 1/\underline{r}$,
$\int_{\mathbf{\Omega}}(g^{k}-f)h\ud\mu\leq x,\,k=1,\ldots,n$,
$\int_{\mathbf{\Omega}}fh\ud\mu\leq \overline{a}/\underline{r}$), we
have
\begin{eqnarray*}
\int_{\mathbf{\Omega}}(f^{n}-f)h\ud\mu & = &
\int_{\mathbf{\Omega}}\left(x_{0} +
\sum_{k=1}^{n}g^{k}\mathbbm{1}_{A_{k}} - f\right)h\ud\mu \\
& \leq & \frac{x_{0}}{\underline{r}} +
\sum_{k=1}^{n}\int_{\mathbf{\Omega}}(g^{k}-f+f)h\ud\mu  -
\int_{\mathbf{\Omega}}fh\ud\mu \\
& = & \frac{x_{0}}{\underline{r}} +
\sum_{k=1}^{n}\int_{\mathbf{\Omega}}(g^{k}-f)h\ud\mu +
(n-1)\int_{\mathbf{\Omega}}fh\ud\mu \\ 
&\leq & \frac{1}{\underline{r}}\left(x_{0}+(n-1)\overline{a}\right) +
nx.  
\end{eqnarray*}
Thus,
$f^{n}\in\mathcal{C}\left(\left(x_{0}+(n-1)\overline{a}\right)/\underline{r}
+ nx\right),\,n\in\mathbb{N}$.

On the other hand, since $U^{+}(\cdot)$ is non-negative and
non-decreasing,
\begin{eqnarray*}
\int_{\mathbf{\Omega}}U(f^{n})\ud\mu = 
\int_{\mathbf{\Omega}}U^{+}(f^{n})\ud\mu
& = & \int_{\mathbf{\Omega}}U^{+}\left(x_{0} +
\sum_{k=1}^{n}g^{k}\mathbbm{1}_{A_{k}}\right)\ud\mu \\
& \geq & \int_{\mathbf{\Omega}}
U^{+}\left(\sum_{k=1}^{n}g^{k}\mathbbm{1}_{A_{k}}\right)\ud\mu \\
& = & \sum_{k=1}^{n}\int_{\mathbf{\Omega}}
U^{+}\left(g^{k}\mathbbm{1}_{A_{k}}\right)\ud\mu \geq \alpha n.
\end{eqnarray*}
Therefore,
\begin{eqnarray*}
\limsup_{z\to\infty}\frac{u(z)}{z} & = &
\limsup_{n\to\infty}\frac{u\left(\left(x_{0} +
(n-1)\overline{a}\right)/\underline{r}
+ nx\right)}{\left(x_{0} + (n-1)\overline{a}\right)/\underline{r} +
nx} \\  
& \geq & \limsup_{n\to\infty}
\frac{\int_{\mathbf{\Omega}}U(f^{n})\ud\mu}{\left(x_{0} +
(n-1)\overline{a}\right)/\underline{r} + nx} \\ 
& \geq & \limsup_{n\to\infty}
\left(\frac{\alpha n}{\left(x_{0} +
(n-1)\overline{a}\right)/\underline{r} + nx}\right)
= \frac{\alpha}{x+\overline{a}/\underline{r}} > 0,
\end{eqnarray*}
which contradicts the limiting weak duality bound in
\eqref{eq:limiting}, and establishes the result.

\end{proof}

We can now prove existence of a unique optimiser in the primal
problem.

\begin{lemma}[Primal existence]
\label{lem:primexis}

The optimal solution $\widehat{g}(x)\in\mathcal{C}(x)$ to the primal
problem \eqref{eq:vfabs} exists and is unique, so that $u(\cdot)$ is
strictly concave.  
  
\end{lemma}

\begin{proof}

Fix $x>0$. Let $(g^{n})_{n\in\mathbb{N}}$ be a maximising sequence
in $\mathcal{C}(x)$ for $u(x)<\infty$ (the finiteness proven in
Lemma \ref{lem:weakdual}). That is
\begin{equation}
\lim_{n\to\infty}\int_{\mathbf{\Omega}}U(g^{n})\ud\mu = u(x) < \infty.
\label{eq:maxseqprimal}
\end{equation}

By the compactness lemma for $\mathcal{C}$ (and thus also for
$\mathcal{C}(x)$), Lemma \ref{lem:Ccompact}, we can find a sequence
$(\widehat{g}^{n})_{n\in\mathbb{N}}$ of convex combinations, so
$\mathcal{C}(x)\owns\widehat{g}^{n}\in \conv(g^{n},
g^{n+1},\ldots),\,n\in\mathbb{N}$, which converges $\mu$-a.e. to some
element $\widehat{g}(x)\in\mathcal{C}(x)$. We claim that
$\widehat{g}(x)$ is the primal optimiser. That is, that we have
\begin{equation}
\int_{\mathbf{\Omega}}U(\widehat{g}(x))\ud\mu = u(x).
\label{eq:primaloptimiser}
\end{equation}
By concavity of $U(\cdot)$ and \eqref{eq:maxseqprimal} we have
\begin{equation*}
\lim_{n\to\infty}\int_{\mathbf{\Omega}}U(\widehat{g}^{n})\ud\mu \geq
\lim_{n\to\infty}\int_{\mathbf{\Omega}}U(g^{n})\ud\mu = u(x),
\end{equation*}
which, combined with the obvious inequality
$u(x)\geq
\lim_{n\to\infty}\int_{\mathbf{\Omega}}U(\widehat{g}^{n})\ud\mu$ means
that we also have, further to \eqref{eq:maxseqprimal},
\begin{equation*}
\lim_{n\to\infty}\int_{\mathbf{\Omega}}U(\widehat{g}^{n})\ud\mu =
u(x).  
\end{equation*}
In other words
\begin{equation}
\lim_{n\to\infty}\int_{\mathbf{\Omega}}U^{+}(\widehat{g}^{n})\ud\mu -
\lim_{n\to\infty}\int_{\mathbf{\Omega}}U^{-}(\widehat{g}^{n})\ud\mu =
u(x) < \infty,    
\label{eq:Uplusminus}
\end{equation}
and note therefore that both integrals in \eqref{eq:Uplusminus} are
finite.

From Fatou's lemma, we have
\begin{equation}
\lim_{n\to\infty}\int_{\mathbf{\Omega}}U^{-}(\widehat{g}^{n})\ud\mu
\geq \int_{\mathbf{\Omega}}U^{-}(\widehat{g}(x))\ud\mu.  
\label{eq:Uminus}
\end{equation}
From Lemma \ref{lem:Uplusg} we have uniform integrability of
$(U^{+}(\widehat{g}^{n}))_{n\in\mathbb{N}}$, so that
\begin{equation}
\lim_{n\to\infty}\int_{\mathbf{\Omega}}U^{+}(\widehat{g}^{n})\ud\mu
= \int_{\mathbf{\Omega}}U^{+}(\widehat{g}(x))\ud\mu.  
\label{eq:Uplus}
\end{equation}
Thus, using \eqref{eq:Uminus} and \eqref{eq:Uplus} in
\eqref{eq:Uplusminus}, we obtain
\begin{equation*}
u(x) \leq \int_{\mathbf{\Omega}}U(\widehat{g}(x))\ud\mu,
\end{equation*}
which, combined with the obvious inequality
$u(x)\geq\int_{\mathbf{\Omega}}U(\widehat{g}(x))\ud\mu$, yields
\eqref{eq:primaloptimiser}. The uniqueness of the primal optimiser
follows from the strict concavity of $U(\cdot)$, as does the strict
concavity of $u(\cdot)$. For this last claim, fix $x_{1}<x_{2}$ and
$\lambda\in(0,1)$, note that
$\lambda\widehat{g}(x_{1}) +
(1-\lambda)\widehat{g}(x_{2})\in\mathcal{C}(\lambda x_{1} +
(1-\lambda)x_{2})$ (yet must be sub-optimal for
$u(\lambda x_{1}+(1-\lambda)x_{2})$ as it is not guaranteed to equal
$\widehat{g}(\lambda x_{1}+(1-\lambda)x_{2})$) and therefore, using
the strict concavity of $U(\cdot)$,
\begin{equation*}
u(\lambda x_{1} + (1-\lambda)x_{2})  \geq
\int_{\mathbf{\Omega}}U\left(\lambda\widehat{g}(x_{1}) +
  (1-\lambda)\widehat{g}(x_{2})\right)\ud\mu > \lambda u(x_{1}) +
(1-\lambda)u(x_{2}). 
\end{equation*}

\end{proof}

We can now move to the dual side of the analysis, which will lead to
the demonstration of conjugacy of the value functions as well as dual
existence and uniqueness. In many duality proofs this is accompanied
by an enlargement of the dual domain, a demonstration of closedness of
the enlarged domain, and subsequent use of the bipolar theorem of
\cite{bs99} on $L^{0}_{+}(\mu)$ to confirm that, with the enlargement,
we have reached the bipolar of the original dual domain. Here, because
the variable $g-f$ appearing in the budget constraint is not an
element of $L^{0}_{+}(\mu)$, the use of the bipolar theorem is not
available, as we noted in Remark \ref{rem:solid}.

Our program is to use the fact that our dual domain is closed, as
established in Lemma \ref{lem:Dclosed}, to directly derive a
compactness lemma for $\mathcal{D}$ (Lemma \ref{lem:Dcompact} below)
and we proceed from there to show dual existence and conjugacy.

\begin{lemma}[Compactness lemma for $\mathcal{D}$]
\label{lem:Dcompact}

Let $(\tilde{h}^{n})_{n\in\mathbb{N}}$ be a sequence in
$\mathcal{D}$. Then there exists a sequence $(h^{n})_{n\in\mathbb{N}}$
with $h^{n}\in\conv(\tilde{h}^{n}, \tilde{h}^{n+1},\ldots)$, which
converges $\mu$-a.e. to an element $h\in\mathcal{D}$ that is
$\mu$-a.e. finite.

\end{lemma}

\begin{proof}

\cite[Lemma A1.1]{ds94} (adapted from a probability space to the
measure space $(\mathbf{\Omega},\mathcal{G},\mu)$) implies the
existence of a sequence $(h^{n})_{n\in\mathbb{N}}$, with
$h^{n}\in\conv(\tilde{h}^{n}, \tilde{h}^{n+1},\ldots)$, which
converges $\mu$-a.e. to an element $h$ that is $\mu$-a.e. finite
because $\mathcal{D}$ is bounded in $L^{0}(\mu)$ (the finiteness also
from \cite[Lemma A1.1]{ds94}). By convexity of $\mathcal{D}$, each
$h^{n},\,n\in\mathbb{N}$ lies in $\mathcal{D}$. Finally, we note that
$h\in\mathcal{D}$ because, according to Lemma \ref{lem:Dclosed},
$\mathcal{D}$ is closed with respect to the topology of convergence in
measure $\mu$.

\end{proof}

The next step in the chain of results we need is a uniform
integrability result for the family
$(V^{-}(h))_{h\in\mathcal{D}(y)}$. The proof uses the
$L^{1}(\mu)$-boundedness of $\mathcal{D}$ and is similar to the proof
in \cite[Lemma 3.2]{ks99}, but the bound on
$\int_{\mathbf{\Omega}}h\ud\mu,\,h\in\mathcal{D}(y)$ here is
$y/\underline{r}$ (as opposed to $y$ in the classical case of
\cite{ks99}). For brevity, therefore, the proof is omitted.

\begin{lemma}[Uniform integrability of $(V^{-}(h))_{h\in\mathcal{D}(y)}$]
\label{lem:Vminush}

The family $(V^{-}(h))_{h\in\mathcal{D}(y)}$ is uniformly integrable,
for any $y>0$.
  
\end{lemma}


  



One can can now proceed to prove existence of a unique optimiser in
the dual problem, and conjugacy of the value functions. We proceed
first with the former, followed by conjugacy.

The proof of Lemma \ref{lem:dualexis} on dual existence is on the same
lines as the proof of primal existence (Lemma \ref{lem:primexis}),
with adjustments for minimisation as opposed to maximisation and
convexity of $V(\cdot)$ (inherited by $V^{(f)}(y):=V(y)+fy,\,y>0$)
replacing concavity of $U(\cdot)$, and uses the uniform integrability
property of $(V^{-}(h))_{h\in\mathcal{D}(y)}$ in Lemma
\ref{lem:Vminush}, so for brevity is omitted.

\begin{lemma}[Dual existence]
\label{lem:dualexis}

The optimal solution $\widehat{h}(y)\in\mathcal{D}(y)$ to the dual
problem \eqref{eq:dvfabs} exists and is unique, so that $v(\cdot)$ is
strictly convex.  
  
\end{lemma}

We can now establish conjugacy of the value functions. The proof works
in the manner of \cite[Lemma 3.4]{ks99}, by bounding the elements in
the primal domain to create a compact set $\mathcal{B}_{n}$ (the set
of elements in $L^{0}_{+}(\mu)$ lying in a ball of radius
$n,\,n\in\mathbb{N}$) for the weak$*$ topology
$\sigma(L^{\infty},L^{1})$ on $L^{\infty}(\mu)$,\footnote{Recall that
  a sequence $(g^{n})_{n\in\mathbb{N}}$ in $L^{\infty}(\mu)$ converges
  to $g\in L^{\infty}(\mu)$ with respect to the weak$*$ topology
  $\sigma(L^{\infty},L^{1})$ if and only if
  $(\langle g^{n},h\rangle)_{n\in\mathbb{N}}$ converges to
  $\langle g,h\rangle$ for each $h\in L^{1}(\mu)$.} so as to apply the
minimax theorem (see \cite[Theorem 45.8]{strasser85}), involving a
maximisation over a compact set and a minimisation over a subset of a
vector space, with the function
$w(g,h):=\int_{\mathbf{\Omega}}(U(g)-gh+fh)\ud\mu =
\int_{\mathbf{\Omega}}(U(g)-(g-f)h)\ud\mu$, for
$g\in\mathcal{B}_{n},\,h\in\mathcal{D}(y)$. The proof uses the dual
characterisation of $\mathcal{C}$ in \eqref{eq:dualC}, along with the
fact that the dual domain $\mathcal{D}(y)$ is bounded in $L^{1}(\mu)$,
as well as the compactness result of Lemma \ref{lem:Dcompact} and the
uniform integrability property of Lemma \ref{lem:Vminush}. As the
proof is on the lines of the proof of \cite[Lemma 3.4]{ks99}, it is
omitted for brevity.






  

  

Note that, because $u(\cdot)$ and $-v(\cdot)$ are strictly
concave, they are almost everywhere differentiable. We shall show in
Lemma \ref{lem:differentiable} that they are in fact differentiable
everywhere, so freely use their derivatives in the statements of some
forthcoming lemmas.


\begin{lemma}[Conjugacy]
\label{lem:conjugacy}
  
The dual value function in \eqref{eq:dvfabs} satisfies the conjugacy
relation
\begin{equation}
v(y) = \sup_{x>0}[u(x)-xy], \quad \mbox{for each $y\in(0,y^{*})$},
\label{eq:absconjugacy}  
\end{equation}
where $u(\cdot)$ is the primal value function in \eqref{eq:vfabs} and
$y^{*}$ is the minimal value of $y>0$ at which the dual derivative
$v^{\prime}(\cdot)$ reaches zero, as in \eqref{eq:y0abs}. 

\end{lemma}

\begin{remark}[Range of validity of \eqref{eq:absconjugacy}]
\label{rem:y0}

The range of $y>0$ for which \eqref{eq:absconjugacy} is valid is
governed by the implicit relation $u^{\prime}(\hat{x})=y$ which
defines the optimal value of $x$ in \eqref{eq:absconjugacy}. We know
that $u(\cdot)$ is increasing and concave and satisfies
$u^{\prime}(\infty)=0$ (from Lemma \ref{lem:uprime0} below). If
$u^{\prime}(0)=+\infty$, then $y^{*}=+\infty$ and
\eqref{eq:absconjugacy} is valid for all $y>0$. We shall see in Lemma
\ref{lem:differentiable} that $y^{*}<\infty$, and this has been
anticipated in \eqref{eq:absconjugacy}.

\end{remark}

We now proceed to further characterise the derivatives of the value
functions, as well as the primal and dual optimisers and the optimal
wealth process.

\begin{lemma}
\label{lem:uprime0}

The derivatives of the primal value function in \eqref{eq:vfabs} at
infinity and of the dual value function in \eqref{eq:dvfabs} at zero
are given by
\begin{equation}
u^{\prime}(\infty) := \lim_{x\to\infty}u^{\prime}(x) = 0, \quad
-v^{\prime}(0) := \lim_{y\downarrow 0}(-v^{\prime}(y)) = +\infty.   
\label{eq:uprime0}
\end{equation}

\end{lemma}

\begin{proof}

By the conjugacy result in Lemma \ref{lem:conjugacy} between the value
functions, the assertions in \eqref{eq:uprime0} are equivalent. We
shall prove the first assertion.

The primal value function $u(\cdot)$ is strictly concave and strictly
increasing, so there is a finite non-negative limit
$u^{\prime}(\infty):=\lim_{x\to\infty}u^{\prime}(x)$. Because
$U(\cdot)$ is increasing with $\lim_{x\to\infty}U^{\prime}(x)=0$, for
any $\epsilon>0$ there exists a number $C_{\epsilon}$ such that
$U(x)\leq C_{\epsilon}+\epsilon x,\,\forall\,x>0$. Using this and the
fact that there exists a positive element $\overline{h}\in\mathcal{D}$
(as in \eqref{eq:boundghbar} in the proof of Proposition
\ref{prop:propertiesCD}) such that
$\int_{\mathbf{\Omega}}g\overline{h}\ud\mu\leq
x(1+\overline{a}/\underline{r}),\,\forall\,g\in\mathcal{C}(x)$, and
l'H\^opital's rule, we have, with
$\int_{\mathbf{\Omega}}\ud\mu=\mathbb{E}\left[\int_{0}^{\infty}
\ud\kappa_{t}\right]=1/\delta>0$,
\begin{eqnarray*}
0 \leq \lim_{x\to\infty}u^{\prime}(x) =
\lim_{x\to\infty}\frac{u(x)}{x} & = &
\lim_{x\to\infty}\sup_{g\in\mathcal{C}(x)}\int_{\mathbf{\Omega}}
\frac{U(g)}{x}\ud\mu \\
& \leq & \lim_{x\to\infty}\sup_{g\in\mathcal{C}(x)}\int_{\mathbf{\Omega}}
\frac{C_{\epsilon}+\epsilon g}{x}\ud\mu \\
& = & \lim_{x\to\infty}\left(\frac{C_{\epsilon}}{\delta x} +
\frac{1}{\overline{h}}\sup_{g\in\mathcal{C}(x)}\int_{\mathbf{\Omega}}
\frac{\epsilon g\overline{h}}{x}\ud\mu\right) \\  
& \leq & \lim_{x\to\infty}\left(\frac{C_{\epsilon}}{\delta x} +
\frac{\epsilon}{\overline{h}}\left(1 +
\frac{\overline{a}}{\underline{r}}\right)\right) 
= \frac{\epsilon}{\overline{h}}\left(1 +
\frac{\overline{a}}{\underline{r}}\right), \quad \mbox{$\mu$-a.e.}, 
\end{eqnarray*}
and taking the limit as $\epsilon\downarrow 0$ gives the result.
  
\end{proof}

The next lemma shows that the dual value function is differentiable
with a derivative that reaches zero at a finite value of its argument,
so that the primal value function has finite derivative at zero.

\begin{lemma}
\label{lem:differentiable}

The dual value function $v(\cdot)$ in \eqref{eq:dvfabs} is
differentiable on $(0,\infty)$. Moreover, as long as the income is
strictly positive, $f>0$, the primal value function in
\eqref{eq:vfabs} has finite derivative at zero:
\begin{equation}
u^{\prime}(0) := \lim_{x\downarrow 0}u^{\prime}(x) < +\infty.
\label{eq:mvprime1}
\end{equation}
Equivalently, the value $y^{*}:=\inf\{y>0:v^{\prime}(y)=0\}$ of $y>0$
at which the derivative of dual value function reaches zero, is
finite.
  
\end{lemma}

\begin{proof}

Since $\int_{\mathbf{\Omega}}h\ud\mu\leq y/\underline{r}$ for any
$h\in\mathcal{D}(y)$, we have
$\int_{\mathbf{\Omega}}(\widehat{h}(y)/y)\ud\mu\leq
1/\underline{r},\,y>0$, which defines a unique integrable element
$L^{0}_{+}(\mu)\supseteq\mathcal{D} \owns
\widehat{h}^{y}:=\widehat{h}(y)/y$, for any $y>0$.

Fix $y>0$. Then, for any $\delta>0$, using the fact that
$(y+\delta)\widehat{h}^{y}$ will be suboptimal for $v(y+\delta)$ along
with convexity of $V(\cdot)$, we have
\begin{eqnarray*}
\frac{1}{\delta}\left(v(y+\delta) - v(y)\right) & \leq &
\frac{1}{\delta}\int_{\mathbf{\Omega}}\left(V((y+\delta)\widehat{h}^{y}) +
f(y+\delta)\widehat{h}^{y} - V(y\widehat{h}^{y}) -
fy\widehat{h}^{y}\right)\ud\mu \\
& \leq &
\int_{\mathbf{\Omega}}\left(V^{\prime}((y+\delta)\widehat{h}^{y})
+ f\widehat{h}^{y}\right)\ud\mu.    
\end{eqnarray*}
The element $(y+\delta)\widehat{h}^{y}\in L^{0}_{+}(\mu)$ is strictly
positive, and thus $|V^{\prime}((y+\delta)\widehat{h}^{y})|$ is
bounded $\mu$-a.e. (on recalling the Inada conditions satisfied by
$-V(\cdot)$), while the non-negative element $f\widehat{h}^{y}$ is
bounded above by the integrable function
$\overline{a}\widehat{h}^{y}$, so we may apply dominated convergence
in sending $\delta\downarrow 0$, to obtain that
\begin{equation}
v^{\prime}(y) \leq
\int_{\mathbf{\Omega}}\left(V^{\prime}(y\widehat{h}^{y}) +
f\right)\widehat{h}^{y}\ud\mu.   
\label{eq:vp1}
\end{equation}

An identical argument, this time applied to
$(v(y)-v(y-\delta))/\delta$, yields the reverse inequality
\begin{equation}
v^{\prime}(y) \geq
\int_{\mathbf{\Omega}}\left(V^{\prime}(y\widehat{h}^{y}) +
f\right)\widehat{h}^{y}\ud\mu.   
\label{eq:vp2}
\end{equation}
Then, \eqref{eq:vp1} and \eqref{eq:vp2} yield that $v(\cdot)$ is
differentiable on $(0,\infty)$ with
\begin{equation}
v^{\prime}(y) =  \int_{\mathbf{\Omega}}\left(V^{\prime}(y\widehat{h}^{y}) +
f\right)\widehat{h}^{y}\ud\mu.
\label{eq:vp}
\end{equation}

It remains to establish \eqref{eq:mvprime1}, and the equivalent
assertion that $y^{*}<\infty$. Now, from the definition of the dual
value function we have
\begin{equation}
v(y) =
\int_{\mathbf{\Omega}}\left(V(y\widehat{h}^{y}) +
yf\widehat{h}^{y}\right)\ud\mu.
\label{eq:mveq1}
\end{equation}
From \eqref{eq:vp} and \eqref{eq:mveq1} we see that the former
is only consistent with the latter if $\widehat{h}^{y}$ has no
explicit dependence on $y$, and this in turn implies (given the
monotonicity of $V^{\prime}(\cdot)$) that the integrand in
\eqref{eq:vp} is monotone in $y$. We can thus apply monotone
convergence to obtain
\begin{equation*}
0 = v^{\prime}(y^{*}) = \lim_{y\uparrow
y^{*}}\int_{\mathbf{\Omega}}\left(V^{\prime}(y\widehat{h}^{y}) +
f\right)\widehat{h}^{y}\ud\mu =
\int_{\mathbf{\Omega}}\left(V^{\prime}(y^{*}\widehat{h}^{y^{*}}) +
f\right)\widehat{h}^{y^{*}}\ud\mu,
\end{equation*}
which in turn yields, since $\widehat{h}^{y^{*}}$ and $f$ are strictly
positive $-V(\cdot)$ satisfies the Inada conditions, that
$y^{*}<\infty$, and the proof is complete.

\end{proof}

\begin{remark}
\label{rem:up}

We shall see the formula \eqref{eq:vp} for the dual derivative
reproduced in the course of proving Lemma \ref{lem:derivatives} (see
\eqref{eq:derivativev}).

The conjugacy between the primal and dual value functions, combined
with Lemma \ref{lem:differentiable}, yields that the primal value
function $u(\cdot)$ is also differentiable on $(0,\infty)$.
  
\end{remark}

The final step in the series of lemmas that will furnish us with the
proof of Theorem \ref{thm:adt} is to obtain a duality characterisation
of the primal and dual optimisers.

\begin{lemma}
\label{lem:derivatives}
  
\begin{enumerate}
  

\item For any fixed $x>0$, with $y=u^{\prime}(x)\in(0,y^{*})$
(equivalently $x=-v^{\prime}(y)$), the primal and dual optimisers
$\widehat{g}(x),\widehat{h}(y)$ are related by
\begin{equation}
U^{\prime}(\widehat{g}(x)) = \widehat{h}(y) =
\widehat{h}(u^{\prime}(x)), \quad \mu\mbox{-a.e.},
\label{eq:primaldual}
\end{equation}
and satisfy
\begin{equation}
\int_{\mathbf{\Omega}}(\widehat{g}(x)-f)\widehat{h}(y)\ud\mu = xy =
xu^{\prime}(x).
\label{eq:saturation}
\end{equation}

\item The derivatives of the value functions satisfy the relations
\begin{eqnarray}
xu^{\prime}(x) & = &
\int_{\mathbf{\Omega}}U^{\prime}(\widehat{g}(x))\left(\widehat{g}(x) -
f\right)\ud\mu, \quad x>0, \label{eq:derivativeu}\\
yv^{\prime}(y) & = & 
\int_{\mathbf{\Omega}}\left(V^{\prime}(\widehat{h}(y)) +
f\right)\widehat{h}(y)\ud\mu, \quad y\in(0,y^{*}).
\label{eq:derivativev}
\end{eqnarray}

\end{enumerate}

\end{lemma}

\begin{proof}

Recall the inequality \eqref{eq:VUbound}, which also applies to the
value functions because they are also conjugate by Lemma
\ref{lem:conjugacy}. We thus have, in addition to \eqref{eq:VUbound},
\begin{equation}
v(y) \geq u(x) -xy, \quad \forall\,x>0,\,y\in(0,y^{*}), \quad
\mbox{with equality iff $y=u^{\prime}(x)$}.
\label{eq:vubound}
\end{equation}
With $\widehat{g}(x)\in\mathcal{C}(x),\,x>0$ and
$\widehat{h}(y)\in\mathcal{D}(y),\,y\in(0,y^{*})$ denoting the primal
and dual optimisers, we have, because
$\int_{\mathbf{\Omega}}(g-f)h\ud\mu\leq xy$ for all
$g\in\mathcal{C}(x),\,h\in\mathcal{D}(y)$,
\begin{equation*}
\int_{\mathbf{\Omega}}(\widehat{g}(x) - f)\widehat{h}(y)\ud\mu \leq
xy, \quad x>0,\,y\in(0,y^{*}).  
\end{equation*}
Using this as well as \eqref{eq:VUbound} and \eqref{eq:vubound} we
have
\begin{eqnarray}
\label{eq:fundineq}
0 & \leq &\int_{\mathbf{\Omega}}\left(V(\widehat{h}(y)) -
U(\widehat{g}(x)) + \widehat{g}(x)\widehat{h}(y)\right)\ud\mu \\
& = &
\int_{\mathbf{\Omega}}\left(V(\widehat{h}(y)) + f\widehat{h}(y) -
U(\widehat{g}(x)) + (\widehat{g}(x) - f)\widehat{h}(y)\right)\ud\mu
\nonumber \\
& \leq & v(y) - u(x) + xy, \quad x>0,\,y\in(0,y^{*}). \nonumber
\end{eqnarray}
The right-hand-side of \eqref{eq:fundineq} is zero if and only if
$y=u^{\prime}(x)$, due to \eqref{eq:vubound}, and the non-negative
integrand must then be $\mu$-a.e. zero, which by \eqref{eq:VUbound}
can only happen if \eqref{eq:primaldual} holds, which establishes that
primal-dual relation.

Thus, for any fixed $x>0$ and with $y=u^{\prime}(x)$, and hence
equality in \eqref{eq:fundineq}, we have
\begin{eqnarray*}
0 & = & \int_{\mathbf{\Omega}}\left(V(\widehat{h}(y)) +
f\widehat{h}(y) - U(\widehat{g}(x)) + (\widehat{g}(x) -
f)\widehat{h}(y)\right)\ud\mu \\ 
& = & v(y) - u(x) +
\int_{\mathbf{\Omega}}(\widehat{g}(x) - f)\widehat{h}(y)\ud\mu \\ 
& = & v(y) - u(x) + xy, \quad y=u^{\prime}(x),
\end{eqnarray*}
which implies that \eqref{eq:saturation} must hold. Inserting the
explicit form of $\widehat{h}(y)=U^{\prime}(\widehat{g}(x))$ into
\eqref{eq:saturation} yields \eqref{eq:derivativeu}. Similarly,
setting $\widehat{g}(x)=I(\widehat{h}(y))=-V^{\prime}(\widehat{h}(y))$
into \eqref{eq:saturation}, with $x=-v^{\prime}(y)$ (equivalent to
$y=u^{\prime}(x)$), yields \eqref{eq:derivativev}.

\end{proof}

We now have all the results needed for the abstract duality in Theorem
\ref{thm:adt}, so let us confirm this.

\begin{proof}[Proof of Theorem \ref{thm:adt}]

Lemma \ref{lem:conjugacy} implies the relations \eqref{eq:conjugacy}
of item (i). The statements in item (ii) are implied by Lemma
\ref{lem:primexis} and Lemma \ref{lem:dualexis}. Items (iii) and (iv)
follow from Lemma \ref{lem:uprime0}, Lemma
\ref{lem:differentiable}, Remark \ref{rem:up} and Lemma
\ref{lem:derivatives}.

\end{proof}

\subsection{Proof of the concrete duality}
\label{subsec:pcd}

We are almost ready to prove the concrete duality in Theorem
\ref{thm:cwrtid}, because Theorem \ref{thm:adt} readily implies nearly
all of the assertions of Theorem \ref{thm:cwrtid}. The outstanding
assertion is the characterisation of the optimal wealth process in
\eqref{eq:owp} and the associated uniformly integrable martingale
property of the deflated wealth plus cumulative deflated consumption
over income process
$\widehat{X}(x)\widehat{Y}(y) +
\int_{0}^{\cdot}(\widehat{c}_{s}(x)-f_{s})\widehat{Y}_{s}(y)\ud s$. So we
proceed to establish these assertions in the proposition below, which
turns out to be interesting in its own right. We take as given the
other assertions of Theorem \ref{thm:cwrtid}, and in particular the
optimal budget constraint in \eqref{eq:obc}. We shall confirm the
proof of Theorem \ref{thm:cwrtid} in its entirety after the proof of
the next result.

\begin{proposition}[Optimal wealth process]
\label{prop:owp}

Given the saturated budget constraint equality in \eqref{eq:obc},
the optimal wealth process is characterised by \eqref{eq:owp}. The
process
\begin{equation*}
\widehat{\Lambda}_{t} := \widehat{X}_{t}(x)\widehat{Y}_{t}(y) +
\int_{0}^{t}(\widehat{c}_{s}(x)-f_{s})\widehat{Y}_{s}(y)\ud s, \quad 0\leq
t <\infty, 
\end{equation*}
is a uniformly integrable martingale, converging to an integrable
random variable $\widehat{\Lambda}_{\infty}$, so the martingale
extends to $[0,\infty]$. The process $\widehat{X}(x)\widehat{Y}(y)$ is
a potential, that is, a non-negative supermartingale satisfying
$\lim_{t\to\infty}\mathbb{E}[\widehat{X}_{t}(x)\widehat{Y}_{t}(y)]=0$.
Moreover, $\widehat{X}_{\infty}(x)\widehat{Y}_{\infty}(y)=0$, almost
surely.

\end{proposition}

\begin{proof}

Recall the saturated budget constraint equality in \eqref{eq:obc}.
It simplifies notation if we take $x=y=1$, and is without loss of
generality: although $y=u^{\prime}(x)$ in \eqref{eq:obc}, one can
always multiply the utility function by an arbitrary constant so as
to ensure that $u^{\prime}(1)=1$.  We thus have the optimal budget
constraint
\begin{equation}
\mathbb{E}\left[\int_{0}^{\infty}(\widehat{c}_{t} -
f_{t})\widehat{Y}_{t}\ud t\right] =1,
\label{eq:obc1}
\end{equation}
for $\widehat{c}\equiv\widehat{c}(1)\in\mathcal{A}$ and
$\widehat{Y}\equiv\widehat{Y}(1)\in\mathcal{Y}$. Since
$\widehat{c}\in\mathcal{A}$, we know there exists an optimal wealth
process $\widehat{X}\equiv\widehat{X}(1)$ and an associated optimal
trading strategy $\widehat{H}$, such that $\widehat{X}\geq 0$ and such
that
$\widehat{\Lambda}:=\widehat{X}\widehat{Y} +
\int_{0}^{\cdot}(\widehat{c}_{s} - f_{s})\widehat{Y}_{s}\ud s$ is a
supermartingale over $[0,\infty)$. The supermartingale condition, by
the same arguments that led to the derivation of the budget constraint
in Lemma \ref{lem:bc}, leads to the inequality
$\mathbb{E}\left[\int_{0}^{\infty}(\widehat{c}_{t} -
  f_{t})\widehat{Y}_{t}\ud t\right] \leq1$ instead of the equality
\eqref{eq:obc1}. Similarly, if the supermartingale is strict, we get a
strict inequality in place of \eqref{eq:obc1}. We thus deduce that
$\widehat{\Lambda}$ must be a martingale over $[0,\infty)$. We shall
show that this extends to $[0,\infty]$, along with the other claims in
the proposition.

Since $\mathcal{Y}\subseteq\mathcal{Y}^{0}$, $\widehat{Y}$ is also a
wealth deflator, so the (non-negative c\`adl\`ag) deflated wealth
process $\widehat{X}\widehat{Y}$ is a is a non-negative c\`adl\`ag
supermartingale, and thus by \cite[Corollary 5.2.2]{ce15} converges to
an integrable limiting random variable
$\widehat{X}_{\infty}\widehat{Y}_{\infty} :=
\lim_{t\to\infty}\widehat{X}_{t}\widehat{Y}_{t}$ (and moreover
$\widehat{X}_{t}\widehat{Y}_{t}\geq
\mathbb{E}[\widehat{X}_{\infty}\widehat{Y}_{\infty}],\,t\geq 0$). The
integral in $\widehat{\Lambda}$ clearly also converges to an
integrable random variable, by virtue of the budget constraint. Thus,
$\widehat{\Lambda}$ also converges to an integrable random variable
$\widehat{\Lambda}_{\infty}:=\widehat{X}_{\infty}\widehat{Y}_{\infty}+
\int_{0}^{\infty}(\widehat{c}_{t}-f_{t})\widehat{Y}_{t}\ud t$. By
\cite[Theorem I.13]{protter}, the extended martingale over
$[0,\infty]$, $(\widehat{\Lambda}_{t})_{t\in[0,\infty]}$ is then
uniformly integrable, as claimed.

The martingale condition gives
\begin{equation*}
\mathbb{E}\left[\widehat{X}_{t}\widehat{Y}_{t} +
\int_{0}^{t}(\widehat{c}_{s}-f_{s})\widehat{Y}_{s}\ud s\right] = 1.
\quad 0\leq t <\infty.
\end{equation*}
Since $\widehat{X}\widehat{Y}$ is non-negative, while the integral on
the left-hand-side is bounded below by an integrable random variable
(due to the bound in \eqref{eq:fbound} on the income stream and the
integrability of $\widehat{Y}$), taking the limit as $t\to\infty$,
using the Fatou lemma and utilising \eqref{eq:obc1} yields
$\lim_{t\to\infty}\mathbb{E}[\widehat{X}_{t}\widehat{Y}_{t}] \leq
0$. But, since $\widehat{X}\widehat{Y}$ is non-negative, we must in
fact have
\begin{equation*}
\lim_{t\to\infty}\mathbb{E}[\widehat{X}_{t}\widehat{Y}_{t}] = 0,
\end{equation*}
so that $\widehat{X}\widehat{Y}$ is a potential, as claimed.

Using the uniform integrability of $\widehat{\Lambda}$ and taking the
limit as $t\to\infty$ in
$\mathbb{E}[\widehat{\Lambda}_{t}]=1,\,t\geq 0$, we have
\begin{equation*}
1 = \lim_{t\to\infty}\mathbb{E}[\widehat{\Lambda}_{t}] =
\mathbb{E}\left[\lim_{t\to\infty}\widehat{\Lambda}_{t}\right] =
\mathbb{E}[\widehat{X}_{\infty}\widehat{Y}_{\infty}] + 1,
\end{equation*}
on using \eqref{eq:obc1}. Hence, we get
$\mathbb{E}[\widehat{X}_{\infty}\widehat{Y}_{\infty}]=0$ and, since
$\widehat{X}_{\infty}\widehat{Y}_{\infty}$ is non-negative, we deduce
that $\widehat{X}_{\infty}\widehat{Y}_{\infty}=0$, almost surely as
claimed.

We can now assemble these ingredients to arrive at the optimal wealth
process formula \eqref{eq:owp}. Applying the martingale condition
again, this time over $[t,u]$ for some $t\geq 0$, we have
\begin{equation*}
\mathbb{E}\left[\left.\widehat{X}_{u}\widehat{Y}_{u} +
\int_{0}^{u}(\widehat{c}_{s}-f_{s})\widehat{Y}_{s}\ud
s\right\vert\mathcal{F}_{t}\right] = \widehat{X}_{t}\widehat{Y}_{t} +
\int_{0}^{t}(\widehat{c}_{s}-f_{s})\widehat{Y}_{s}\ud
s, \quad 0\leq t\leq u <\infty.
\end{equation*}
Taking the limit as $u\to\infty$ and using the uniform integrability
of $\widehat{\Lambda}$ we obtain
\begin{equation*}
\mathbb{E}\left[\left.\lim_{u\to\infty}\left(\widehat{X}_{u}\widehat{Y}_{u} +
\int_{0}^{u}(\widehat{c}_{s}-f_{s})\widehat{Y}_{s}\ud
s\right)\right\vert\mathcal{F}_{t}\right] = \widehat{X}_{t}\widehat{Y}_{t} +
\int_{0}^{t}(\widehat{c}_{s}-f_{s})\widehat{Y}_{s}\ud
s, \quad t\geq 0,
\end{equation*}
which, on using $\widehat{X}_{\infty}\widehat{Y}_{\infty}=0$,
re-arranges to \eqref{eq:owp}, and the proof is complete.
  
\end{proof}

We can now complete the proof of the concrete duality theorem.

\begin{proof}[Proof of Theorem \ref{thm:cwrtid}]

Given the definitions of the sets $\mathcal{C}(x)$ and
$\mathcal{D}(y)$ in \eqref{eq:Cx} and \eqref{eq:Dy}, respectively, and
the identification of the abstract value functions in \eqref{eq:vfabs}
and \eqref{eq:dvfabs} with their concrete counterparts in
\eqref{eq:pvf} and \eqref{eq:dvf}, Theorem \ref{thm:adt} implies all
the assertions of Theorem \ref{thm:cwrtid}, with the exception of the
optimal wealth process formula \eqref{eq:owp} and the uniformly
integrable martingale property of $\widehat{X}(x)\widehat{Y}(y) +
\int_{0}^{\cdot}(\widehat{c}_{s}(x)-f_{s})\widehat{Y}_{s}(y)\ud s$,
which are established by Proposition \ref{prop:owp}.

\end{proof}

It remains to prove Proposition \ref{prop:dense}, from which Theorem
\ref{thm:dense} will follow.

\begin{proof}[Proof of Proposition \ref{prop:dense}]

With $x=y=1$, from Lemma \ref{lem:bc} we know that the budget
constraint \eqref{eq:bc} holds for all $Y\in\widetilde{\mathcal{Y}}$
and all $c\in\mathcal{A}$, so we have the implication
\begin{equation}
Y\in\widetilde{\mathcal{Y}} \implies
\mathbb{E}\left[\int_{0}^{T}(c_{t} - f_{t})Y_{t}\ud t\right]\leq 1, \quad
\forall\,c\in\mathcal{A}.
\label{eq:forwardY}
\end{equation}
Invoking the financing condition of Asumption \ref{ass:financing} with
$Y\in\widetilde{\mathcal{Y}}$, we have the reverse
implication. Translating the resulting equivalence into the abstract
notation on the measure space $(\mathbf{\Omega},\mathcal{G},\mu)$, we
have the analogue of \eqref{eq:dualC}, with $\widetilde{\mathcal{D}}$
in place of $\mathcal{D}$:
\begin{equation*}
g\in\mathcal{C} \iff \langle g-f,h\rangle \leq 1, \quad \forall\,
h\in\widetilde{\mathcal{D}}.  
\end{equation*}

An examination of the abstract duality proof shows that it was crucial
to establish that the abstract dual domain $\mathcal{D}$ was closed,
as in Lemma \ref{lem:Dclosed}, and that this was done via
supermartingale convergence results, where we found a supermartingale
$Y$ such that deflated wealth plus cumulative deflated consumption
over income, $XY+\int_{0}^{\cdot}(c_{s}-f_{s})\ud s$, was also a
supermartingale, so we could conclude that $Y\in\mathcal{Y}$ (because
of the definition of $\mathcal{Y}$) and hence that
$h=\zeta Y\in\mathcal{D}$. This argument would fail with
$\widetilde{\mathcal{D}}$ in place of $\mathcal{D}$, because the
limiting supermartingale in the Fatou convergence argument is known
only to be a supermartingale, and cannot be shown to be a discounted
local martingale deflator $Y\in\widetilde{\mathcal{Y}}$.

However, if we enlarge $\widetilde{\mathcal{D}}$
to its closure, this is automatically closed, and if we define an abstract
dual value function by
\begin{equation*}
\tilde{v}(y) :=
\inf_{h\in\cl(\widetilde{\mathcal{D}}(y))}\int_{\mathbf{\Omega}}\left(V(h)
+ fh\right)\ud\mu, \quad y>0,  
\end{equation*}
then the rest of the abstract duality proof goes through unaltered,
and by conjugacy of the abstract primal and dual value functions, the
dual function $\tilde{v}(\cdot)$ coincides with $v(\cdot)$. The dual
minimiser lies in $\cl(\widetilde{\mathcal{D}}(y))$, and this
establishes the proposition.

\end{proof}

\begin{proof}[Proof of Theorem \ref{thm:dense}]

Furnished with Proposition \ref{prop:dense}, we have established the
result, due to the one-to-one correspondence between
$\widetilde{\mathcal{D}}$ and $\widetilde{\mathcal{Y}}$.
  
\end{proof}

\subsection{The finite horizon case}
\label{subsec:finiteh}

Our duality results work equally well for finite horizon versions of
the problem \eqref{eq:pvf}, with or without a terminal wealth
objective, by making suitable adjustments to the measure $\kappa$, as
indicated in the examples which follow.

\begin{example}[The finite horizon pure consumption problem]
\label{examp:fhcp}
    
With $T\in(0,\infty)$ a fixed terminal date, one has the finite
horizon version of \eqref{eq:pvf}, without a terminal wealth
objective:
$u(x) := \sup_{c\in\mathcal{A}(x)}\mathbb{E}\left[
  \int_{0}^{T}U(c_{t})\ud\kappa_{t}\right]$, where $\kappa$ is again
discounted Lebesgue measure, this time over $[0,T]$. The budget
constraint and dual problem are then the same as in \eqref{eq:bc} and
\eqref{eq:dvf}, with a finite upper integration limit set to $T$, and
once again $\mu:=\kappa\times\mathbb{P}$ is the product measure on the
abstract formulation of the problems.

An examination of the duality proof in the infinite horizon problem
shows that the boundedness in $L^{1}(\mu)$ of the dual domain, and
hence the finiteness of the integrals in Lemma \ref{lem:Dbounded}, was
used on numerous occasions. In the finite horizon case, it is easier
to obtain corresponding results, regardless of whether the interest
rate is strictly positive. With $h=\zeta Y\in\mathcal{D}$, we obtain
$\int_{\mathbf{\Omega}}h\ud\mu \leq T <\infty$, as the reader can
verify.
The finite horizon pure consumption problem thus affords an easier
duality proof. The results of Theorem \ref{thm:cwrtid} hold with the
upper limit of time integration set to $T<\infty$, and the optimal
deflated wealth plus cumulative deflated consumption over income is
simply a uniformly integrable martingale over $[0,T]$.

We observe that we can in fact write the problem as an infinite
horizon problem provided we define the measure $\kappa$
appropriately. So, if we set
\begin{equation*}
\kappa_{t} = \frac{1}{\delta}\left((1-\exp(-\delta t))\mathbbm{1}_{[0,T)}(t)
+ (1-\exp(-\delta T))\mathbbm{1}_{[T,\infty)}(t)\right), \quad \zeta_{t} =
\exp(\delta t), \quad t\geq 0,    
\end{equation*}
then the primal and dual value functions, as well as the budget
constraint, take the same form as in the infinite horizon problem (so
the infinite horizon duality proof actually goes through without
alteration). This observation will be useful in the example which
follows.

\end{example}

\begin{example}[The finite horizon consumption and terminal wealth
  problem]
\label{examp:fhctwp}

In a similar manner to Example \ref{examp:fhcp} we can consider the
problem with a terminal wealth objective:
$u(x) := \sup_{c\in\mathcal{A}(x)}\mathbb{E}\left[
\int_{0}^{T}U(c_{t})\ud\kappa_{t} + U(X_{T})\right]$, for which the
budget constraint is
$\mathbb{E}\left[X_{T}Y_{T} + \int_{0}^{T}(c_{t} - f_{t})Y_{t}\ud
  t\right]$ and the dual problem is
$v(y) := \inf_{Y\in\mathcal{Y}(y)}\mathbb{E}\left[
\int_{0}^{T}\left(V(\zeta_{t}Y_{t}) +
f_{t}\zeta_{t}Y_{t}\right)\ud\kappa_{t} + V(Y_{T})\right]$.

As in the last part of Example \ref{examp:fhcp}, if we define
$\kappa,\zeta$, and also $f$, appropriately, we can express the
problems and budget constraint as infinite horizon problems, for which
the duality results will hold. Thus, if we write
\begin{equation*}
\kappa_{t} = \frac{1}{\delta}(1-\exp(-\delta t))\mathbbm{1}_{[0,T)}(t)
+ \left(\frac{1}{\delta}(1-\exp(-\delta T)) +
1\right)\mathbbm{1}_{[T,\infty)}(t), \quad t\geq 0,       
\end{equation*}
along with
\begin{equation*}
  \zeta_{t} = \left\{\begin{array}{ll}\exp(\delta t), & t\in[0,T), \\
1, & t\in[T,\infty), \end{array}\right. \quad f_{t}=
a_{t}N_{t}\mathbbm{1}_{\{t<T\}}, \quad t\geq 0, 
\end{equation*}
then then the primal and dual value functions, as well as the budget
constraint, take the same form as in the infinite horizon problem.

\end{example}

\section{Analysis of the Davis-Vellekoop example}
\label{sec:dve}


In this section we consider and numerically solve the example of
\cite{dv09}, as described in Section \ref{subsec:canonical}. We first
explore ramifications of the duality results, before a numerical
solution of the pre-income termination HJB equation, when adopting a
stochastic control approach.

The wealth dynamics are
\begin{equation}
\ud X_{t} =  (rX_{t} - c_{t} + aN_{t})\ud t + \sigma\pi_{t}(\lambda \ud
t + \ud W_{t}), \quad X_{0}=x,
\label{eq:dvdyn}
\end{equation}
where $\pi:=HS$ is the wealth invested in the stock, and we have
recorded in \eqref{eq:dvdyn} that the income rate $f$ is given by
$f_{t} = aN_{t},\,t\geq 0$, with $a\geq 0$ constant.

By Theorem \ref{thm:dense}, the space of consumption deflators
coincides with the closure of the space of discounted local martingale
deflators, so without loss of generality we use deflators given by
\begin{equation}
Y_{t} = y\e^{-rt}\mathcal{E}(-\lambda W-\gamma\cdot M)_{t} =
y\e^{-rt}Z_{t} = y\e^{-rt}Z^{(0)}_{t}\mathcal{E}(-\gamma\cdot M)_{t}, 
\quad y>0, \quad t\geq 0,
\label{eq:Ygamma}
\end{equation}
for c\`agl\`ad adapted processes $\gamma$ satisfying
$\gamma>-1,\,\gamma<+\infty$ almost surely, where $M$ the martingale
in \eqref{eq:M}, $Z\in\mathcal{Z}$ is a local martingale deflator and
$Z^{(0)}:=\mathcal{E}(-\lambda W)_{t}$ is the martingale deflator in
the absence of the income, that is, in the underlying Black-Scholes
market.


The supermartingale of \eqref{eq:supermartingale} is given as
\begin{equation}
\Lambda_{t} := X_{t}Y_{t} + \int_{0}^{t}(c_{s}-f_{s})Y_{s}\ud s = xy +
\int_{0}^{t}Y_{s}(\sigma\pi_{s} - \lambda X_{s})\ud W_{s} -
\int_{0}^{t}X_{s-}Y_{s-}\gamma_{s}\ud M_{s}, \quad t\geq 0. 
\label{eq:Lambda}
\end{equation}

We take $U(\cdot)$ to be a power utility: 
$U(x) = \frac{x^{p}}{p}$, $x\geq 0$, $p\in(0,1)$.
The primal value function is defined as in \eqref{eq:pvf}.
The convex conjugate of $U(\cdot)$ is
$V(y):=-y^{q}/q,\,y>0$, $q:=-p/(1-p)$.
Given the structure of the deflators in \eqref{eq:Ygamma}, the dual
value function is given by
\begin{equation*}
v(y) :=
\inf_{Z\in\mathcal{Z}}\mathbb{E}\left[\int_{0}^{\infty}\e^{-\delta
t}\left(V(yZ_{t}\e^{(\delta-r)t}) +
yf_{t}Z_{t}\e^{(\delta-r)t}\right)\ud t\right], \quad y>0.
\end{equation*}
Assume the dual minimiser is given  by
$\widehat{Z} :=\mathcal{E}(-\lambda W - \widehat{\gamma}\cdot M)$, for
some optimal integrand $\widehat{\gamma}$ in \eqref{eq:Ygamma}. For
use below, define the non-negative martingales $J,F$ by
\begin{equation*}
J_{t} := \mathbb{E}\left[\left.
\int_{0}^{\infty}\e^{-qrs-\delta(1-q)s}\widehat{Z}^{q}_{s}\ud
s\right\vert\mathcal{F}_{t}\right], \quad F_{t} := \mathbb{E}\left[\left.
\int_{0}^{\infty}\e^{-rs}f_{s}\widehat{Z}^{q}_{s}\ud
s\right\vert\mathcal{F}_{t}\right] \quad t\geq 0. 
\end{equation*}

Using Theorem \ref{thm:cwrtid}, and in particular
\eqref{eq:pdc}, the optimal consumption process is given by
\begin{equation}
(\widehat{c}_{t}(x))^{-(1-p)} = u^{\prime}(x)\e^{(\delta
- r)t}\widehat{Z}_{t}, \quad t\geq 0. 
\label{eq:chat}
\end{equation}
By \eqref{eq:obc} the optimisers satisfy the saturated budget
constraint
\begin{equation}
\mathbb{E}\left[\int_{0}^{\infty}(\widehat{c}_{t}(x) -
  f_{t})\e^{-rt}\widehat{Z}_{t}\ud t\right] = x.
\label{eq:sbc}
\end{equation}
The relations \eqref{eq:chat} and \eqref{eq:sbc} yield
\begin{equation}
\widehat{c}_{t}(x) = \left(\frac{x +
F_{0}}{J_{0}}\right)\e^{-(\delta-r)(1-q)t}\widehat{Z}^{-(1-q)}_{t},
\quad t\geq 0. 
\label{eq:con}
\end{equation}
Using \eqref{eq:owp}, the optimal wealth process is given by
\begin{equation*}
\e^{-rt}\widehat{X}_{t}(x)\widehat{Z}_{t} =
\left(\frac{x+F_{0}}{J_{0}}\right)\mathbb{E}\left[\left.\int_{t}^{\infty}
\e^{-qrs-\delta(1-q)s}\widehat{Z}^{q}_{s}\ud
s\right\vert\mathcal{F}_{t}\right], \quad t\geq 0.  
\end{equation*}
More pertinently, the optimal martingale $\widehat{\Lambda}$,
corresponding to the process in \eqref{eq:Lambda} at the optimum, is
computed as
\begin{equation*}
\widehat{\Lambda}_{t} := \e^{-rt}\widehat{X}_{t}(x)\widehat{Z}_{t} +
\int_{0}^{t}(\widehat{c}_{s}(x)-f_{s})\e^{-rs}\widehat{Z}_{s}\ud s =
\left(\frac{x+F_{0}}{J_{0}}\right)J_{t} - F_{t}, \quad t\geq 0,
\end{equation*}
so is indeed a martingale.

By martingale representation, $\widehat{\Lambda}$ will have a
stochastic integral representation which, without loss of generality,
can be written in the form
\begin{equation*}
\widehat{\Lambda}_{t} = x +
\int_{0}^{t}e^{-rs}\widehat{Z}_{s}\widehat{X}_{s}(x)(\varphi_{s}
-q\lambda)\ud W_{s} +
\int_{0}^{t}\e^{-rs}\widehat{Z}_{s-}\widehat{X}_{s-}(x)\beta_{s}\ud
M_{s}, \quad t\geq 0,  
\end{equation*}
for some integrands $\varphi,\beta$. Comparing with the representation
in \eqref{eq:Lambda} at the optimum yields the optimal trading strategy in
terms of the optimal portfolio proportion
$\widehat{\theta}:=\widehat{\pi}/\widehat{X}(x)$ and the optimal
integrand $\widehat{\gamma}$ in the form
\begin{equation}
\widehat{\theta}_{t} := \frac{\widehat{\pi}_{t}}{\widehat{X}_{t}(x)} =
\frac{\lambda}{\sigma(1-p)} + \frac{\varphi_{t}}{\sigma}, \quad
\widehat{\gamma}_{t} = -\beta_{t}, \quad t\geq 0.
\label{eq:inv}  
\end{equation}
In particular, the process $\varphi$ records the correction to the
Merton-type strategy $\lambda/(\sigma(1-p))$.

This is as far as one can go without computing explicitly the dual
minimiser $\widehat{Z}$, which is impossible in closed form. We thus
turn to a numerical solution of the primal pre-income termination
problem via the associated HJB equation. For this, we first state the
results for the no-income and perpetual income problems.

\subsection{No-income and perpetual income cases}
\label{subsec:special}

There are well-known closed form solutions for the special cases where
there is no income ($f\equiv 0$) or the income is perpetual
($f_{t}=a,\,t\geq 0$).

In the case with \textit{perpetual income}, define the risk-adjusted
value of the perpetual income stream given by
$\mathbb{E}\left[\int_{0}^{\infty}\e^{-rt}f_{t}Z^{(0)}_{t}\ud
  t\right]=a/r$. It is well known that (see \cite{dv09} for example)
the value function for the non-terminating income problem is
$u_{\infty}:(-a/r,\infty)\to\mathbb{R}$, given by
\begin{equation*}
u_{\infty}(x) = u_{0}\left(x + \frac{a}{r}\right), \quad x+\frac{a}{r} >0. 
\end{equation*}
where $u_{0}:\mathbb{R}_{+}\to\mathbb{R}$ is the Merton no-income
value function, given by
\begin{equation}
u_{0}(x) =  K^{-(1-p)}\frac{x^{p}}{p}, \quad K :=
\frac{1}{1-p}\left(\delta - rp + \frac{1}{2}q\lambda^{2}\right). 
\label{eq:mertonvf}
\end{equation}
A necessary and sufficient condition for a well-posed problem is
$K>0$.

The optimal consumption and investment processes in the perpetual
income problem are given by
\begin{equation}
c^{(\infty)}_{t} = K\left(X^{(\infty)}_{t} + \frac{a}{r}\right), \quad
\pi^{(\infty)}_{t} = \frac{\lambda}{\sigma(1-p)}\left(X^{(\infty)}_{t} +
\frac{a}{r}\right), \quad t\geq 0,
\label{eq:mertoncpi}    
\end{equation}
where $X^{(\infty)}$ is the optimal wealth process, given by
\begin{equation}
X^{(\infty)}_{t} + \frac{a}{r} = \left(x + \frac{a}{r}\right)
\e^{(r-\delta)(1-q)t}\left(Z^{(0)}_{t}\right)^{-(1-q)}, \quad
t\geq 0.   
\label{eq:mertonX}
\end{equation}
The optimal strategies $c^{(0)},\pi^{(0)}$ and wealth process $X^{(0)}$
in the no-income problem are recovered by letting $a\downarrow 0$ in
the solutions for the perpetual income case.

Naturally, the value function for the non-terminating income problem
is defined for negative values of initial wealth. This is the concrete
manifestation of the fact that the agent borrows against the known
present value of future income, and implements the no-income optimal
strategy with the increased initial capital $x+a/r$. It is precisely
this strategy that is not available to the agent when the income
terminates randomly.

\subsection{The terminating income case}
\label{subsec:tic}

In the case with randomly terminating income, we can immediately bound
the achievable utility between that for the problems with no income
and perpetual income, so the primal value function $u(\cdot)$ in
\eqref{eq:pvf} satisfies the bounds
\begin{equation}
u_{0}(x)\leq u(x)\leq u_{0}\left(x+\frac{a}{r}\right), \quad x > 0.
\label{eq:npb}  
\end{equation}
In particular, $u(0)$ is finite.

Furthermore, the bounds in \eqref{eq:npb} imply that, as $x\to\infty$,
the value functions $u_{0}(\cdot)$, $u(\cdot)$ and $u_{\infty}(\cdot)$ all
coalesce at large values of wealth, as do their derivatives, and the
Inada conditions for $u_{0}(\cdot)$ at infinity yield
that the derivative of the primal value function at infinity is given by
\begin{equation}
u^{\prime}(\infty) := \lim_{x\to\infty}u^{\prime}(x) = 0,
\label{eq:upinfinity}
\end{equation}
in accordance with our earlier duality results.

More pertinently, we thus have an approximation for the primal value
function at large values of initial wealth:
\begin{equation}
u_{0}(x) \approx u(x) \approx u_{\infty}(x), \quad \mbox{as
$x\to\infty$}. 
\label{eq:approx}
\end{equation}
This approximation is useful in a numerical solution of the primal HJB
equation for the pre-income termination component of the value
function, as we describe shortly.

\subsubsection{HJB equation for the pre-income termination
problem} 
\label{subsubsec:hjbeqn}

It is manifestly the case that the agent will adopt the Merton
no-income optimal strategy as soon as the income terminates, as was
observed by \cite{dv09}. Thus, the maximal expected utility process
must satisfy
\begin{equation}
u(\widehat{X}_{t}) = N_{t}u_{1}(\widehat{X}_{t}) +
(1-N_{t})u_{0}(\widehat{X}_{t}), \quad t\geq 0,  
\label{eq:eup}
\end{equation}
where $\widehat{X}$ is the optimal wealth process, and where
$u_{1}:\mathbb{R}_{+}\to\mathbb{R}$ is the value function of the
\textit{random horizon control problem}
\begin{equation}
u_{1}(x) = \sup_{(\pi,c)\in\mathcal{H}(x)}\mathbb{E}\left[
\int_{0}^{\tau}U(c_{t})\ud\kappa_{t} +
\exp(-\delta\tau)u_{0}(X_{\tau})\right], \quad x>0,
\label{eq:rhcp}
\end{equation}
subject to dynamics over $[0,\tau)$ given by \eqref{eq:dvdyn}, so that
in particular, the income rate is the constant $a\geq 0$ up to the
termination time.

The optimal consumption process must decompose according to
\begin{equation*}
\widehat{c}_{t} = N_{t}c^{(1)}_{t} +
(1-N_{t})c^{(0)}(\widehat{X}_{t}), \quad t\geq 0,  
\end{equation*}
where $c^{(0)}(\widehat{X})=K\widehat{X}$ is the optimal Merton
no-income feedback control function, and $c^{(1)}$ denotes the optimal
consumption process up to the random time at which the income
terminates, and we expect to have an associated feedback control
function $c^{(1)}(\cdot)$, associated with the HJB equation for the
random horizon control problem in \eqref{eq:rhcp}, such that
$c^{(1)}_{t}=c^{(1)}(\widehat{X}_{t}),\,t\in[0,\tau)$.
Similar remarks apply to the optimal investment strategy.

One can write down the HJB equation associated with the value
functions $u_{1}(\cdot)$, by examining the expected utility process in
\eqref{eq:eup}, requiring it to be a supermartingale for any
admissible control and a martingale for the optimal control. This
would yield the HJB equation associated with $u_{1}(\cdot)$ (as given
in \cite{dv09})
\begin{equation}
\sup_{(\pi,c)}\left(U(c) + (rx - c + a +
\sigma\lambda\pi)u_{1}^{\prime}(x) +
\frac{1}{2}\sigma^{2}\pi^{2}u_{1}^{\prime\prime}(x) - \delta u_{1}(x)
+ \eta(u_{0}(x) - u_{1}(x))\right) = 0.
\label{eq:HJBti}
\end{equation}

\subsubsection{Transformation of the pre-termination control
  problem}
\label{subsubsec:transformation}

An equivalent way to arrive at \eqref{eq:HJBti} is to begin from the
random horizon formulation in \eqref{eq:rhcp} and integrate over the
distribution of $\tau$, using its density function and the
independence of $\tau$ from the Brownian motion $W$. After performing
an integration by parts in the term involving the integral of
consumption, we arrive at an infinite horizon control problem to
maximise \textit{utility from consumption and inter-temporal wealth},
and with a perpetual (so, non-terminating) income stream paying at the
constant rate $a\geq 0$. In other words, the value function
$u_{1}(\cdot)$ is expressed in the form
\begin{equation}
u_{1}(x) = \sup_{(\pi,c)\in\mathcal{H}(x)}\mathbb{E}\left[
\int_{0}^{\infty}\exp(-\alpha t)\left(U(c_{t}) +
U_{0}(X_{t})\right)\ud t\right], \quad x>0,
\label{eq:ihitwp}
\end{equation}
with a modified discount factor $\alpha$ and a utility function
$U_{0}:\mathbb{R}_{+}\to\mathbb{R}$ (measuring felicity from
inter-temporal wealth) given respectively by
\begin{equation*}
\alpha := \eta +\delta, \quad U_{0}(\cdot) := \eta u_{0}(\cdot),  
\end{equation*}
with $u_{0}(\cdot)$ the Merton no-income value function, and
\eqref{eq:ihitwp} is subject to wealth dynamics over $[0,\infty)$
given by
\begin{equation*}
\ud X_{t} =  (rX_{t} - c_{t} + a)\ud t + \sigma\pi_{t}(\lambda \ud
t + \ud W_{t}), \quad X_{0}=x,  
\end{equation*}
so of the form in \eqref{eq:dvdyn}, but with no termination of the
income (since we have integrated over all possible values of the
termination time). The HJB equation for the problem in
\eqref{eq:ihitwp} is indeed that in \eqref{eq:HJBti}.

Performing the maximisation in the HJB equation \eqref{eq:HJBti}
yields the optimal feedback control functions
$c^{(1)}(\cdot),\pi^{(1)}(\cdot)$, given by
\begin{equation}
\label{eq:ofc}
c^{(1)}(x) = I(u_{1}^{\prime}(x)), \quad \pi^{(1)}(x) =
-\frac{\lambda}{\sigma}\frac{u_{1}^{\prime}(x)}{u_{1}^{\prime\prime}(x)},
\end{equation}
where $I(\cdot):=(U^{\prime}(\cdot))^{-1}$.


Substituting the feedback control functions the into the Bellman
equation converts it to the non-linear ODE
\begin{equation}
V(u_{1}^{\prime}(x)) + U_{0}(x) + (rx+a)u_{1}^{\prime}(x) -
\frac{1}{2}\lambda^{2}\frac{(u_{1}^{\prime}(x))^{2}}{u_{1}^{\prime\prime}(x)}
- \alpha u_{1}(x) = 0, \quad x\in\mathbb R_{+},
\label{eq:primalode}
\end{equation}
where $V(\cdot)$ is the convex conjugate of $U(\cdot)$. The ODE
\eqref{eq:primalode} has no closed form solution.

HJB equations for problems of utility from consumption and
inter-temporal wealth have been considered by \cite{fgg15}, who proved
regularity of solutions to HJB equations for problems of the sort in
\eqref{eq:ihitwp} with a finite horizon, so
the HJB equations were PDEs as opposed to ODEs, but \cite{fgg15}
remark that the extension to the infinite horizon is not problematic,
and this is reasonable. The other additional ingredient in
\eqref{eq:primalode} and the problem \eqref{eq:ihitwp} compared with
\cite{fgg15} is the presence of the constant income stream, resulting
in the additional linear term $au_{1}^{\prime}(\cdot)$ in the
ODE. This renders impossible the closed form solution of
\eqref{eq:primalode}, but we conjecture that the regularity of the
solution would not be affected, so in this section we make the
assumption that:

\begin{assumption}
\label{ass:smooth}
  
There is a smooth (that is $C^{2}(\mathbb{R}_{+})$) solution to the
HJB equation \eqref{eq:primalode} for power utility with power
$p\in(0,1)$.
  
\end{assumption}

Indeed, we know from the duality theory developed in earlier sections
that the value function of the terminating income problem is
differentiable in the wealth argument, so at worst the second
derivative would exist as a distribution. For these reasons, and in
order to present a numerical solution to the pre-income termination
function $u_{1}(\cdot)$, we invoke Assumption \ref{ass:smooth} in this
section.


\subsubsection{Numerical illustration}
\label{subsubsec:numerics}

We implemented an explicit Runge-Kutta method to solve the primal HJB 
ODE \eqref{eq:primalode}. This method is an extension of the standard 
Euler method, using a specific weighted sum of increments of the value 
function $u_{1}(\cdot)$ in a spatial discretisation to increase the order of 
convergence of the algorithm. The method was initiated at a large value 
of wealth $\overline{x}$ such that the approximation in \eqref{eq:approx} 
was operative, so for the pre-income termination value function we used 
$u_{1}(\overline{x})\approx u_{0}\left(\overline{x} + a/(r + \eta)\right)$.
This allowed for a computation of the first two derivatives of $u_{1}(\cdot)$
at large values of wealth, and then the Runge-Kutta algorithm
proceeds to compute the value function at all values of $x$ down to
zero in the wealth grid.

\begin{figure}[htb]
  
\begin{center}
\begin{minipage}{.5\textwidth}
\centering
\begin{subfigure}[b]{\textwidth}
\includegraphics[width=\textwidth]{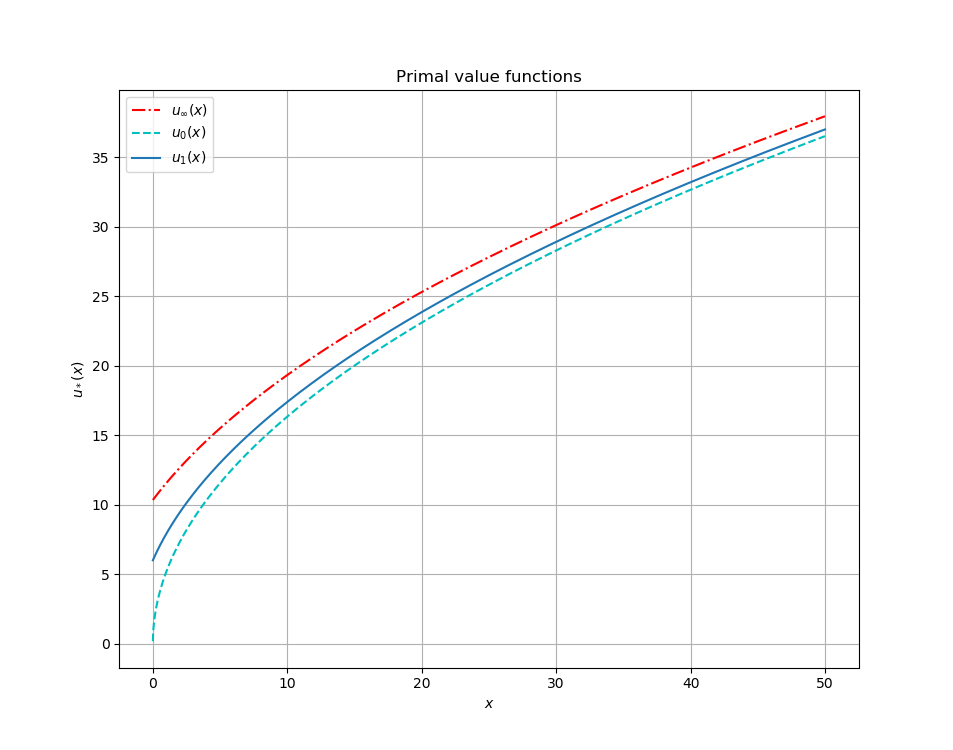}
\caption{Pre-income termination value function.}
\end{subfigure}
\end{minipage}%
\begin{minipage}{.5\textwidth}
\centering
\begin{subfigure}[b]{\textwidth}
\includegraphics[width=\textwidth]{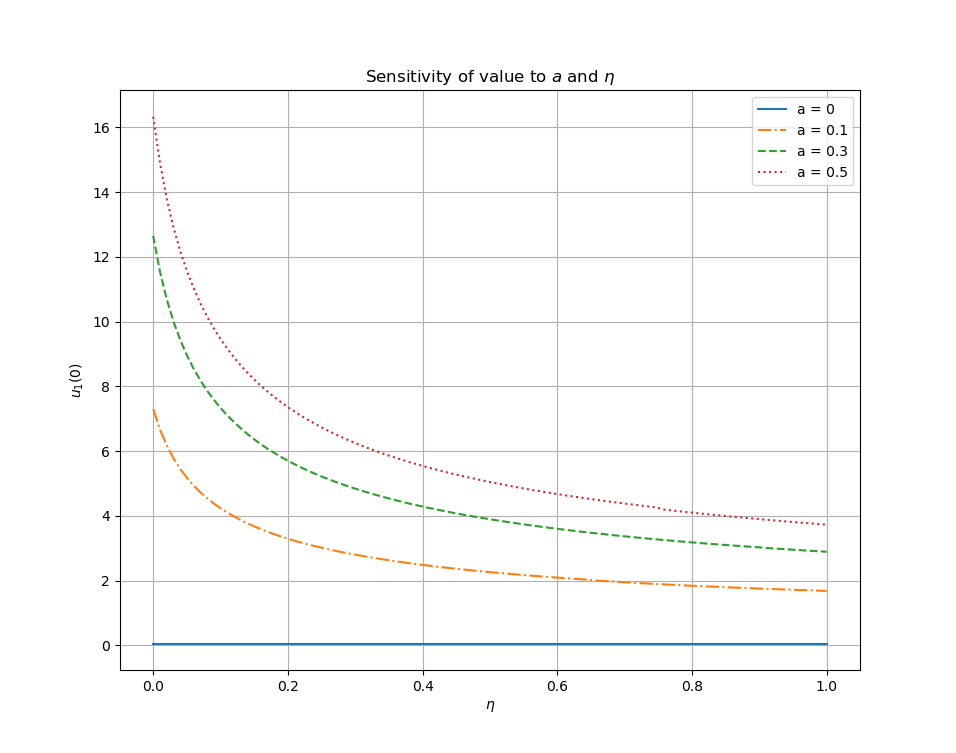}
\caption{Sensitivity of value function to $a$ and $\eta$.} 
  \end{subfigure}
\end{minipage}
\caption{Pre-income termination value function and the no-income and
  perpetual income value functions, and sensitivity of the value
  function to the income rate $a$ and termination intensity~$\eta$ at
  zero wealth. }
\label{fig:pvf}

\end{center}

\end{figure}

Figure \ref{fig:pvf} shows the value function $u_{1}(\cdot)$ as well
as its bounding functions $u_{0}(\cdot)$ and $u_{\infty}(\cdot)$, also
the dependence of the value function on the income rate $a$ and
termination intensity $\eta$ at zero wealth.  The data used are
$\lambda=1$, $\sigma=0.1$, $\delta=0.6$, $\eta=0.1$, $a=0.2$, $p=0.5$.

\begin{figure}[htb]
  
\begin{center}
\begin{minipage}{.5\textwidth}
\centering
\begin{subfigure}[b]{\textwidth}
\includegraphics[width=\textwidth]{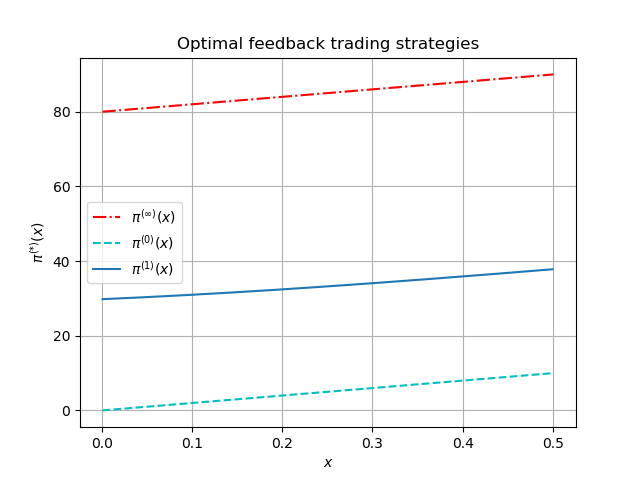}
\caption{Optimal feedback investment.}
\end{subfigure}
\end{minipage}%
\begin{minipage}{.5\textwidth}
\centering
\begin{subfigure}[b]{\textwidth}
\includegraphics[width=\textwidth]{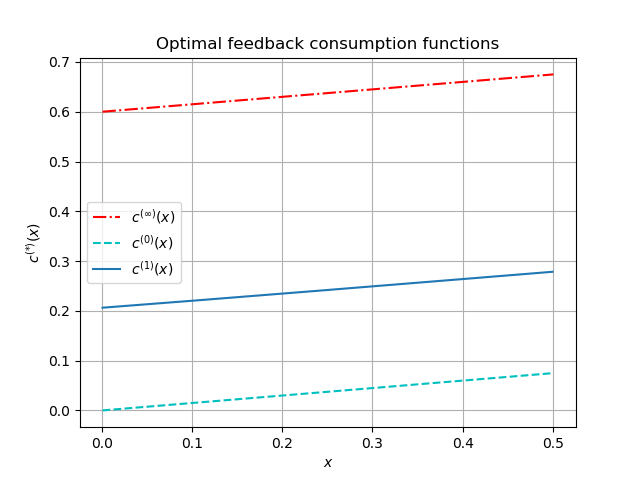}
\caption{Optimal feedback consumption.} 
  \end{subfigure}
\end{minipage}
\caption{Optimal feedback control functions associated with the
  pre-income termination value function $u_{1}(\cdot)$, and the
  no-income and perpetual income value functions $u_{0}(\cdot)$ and
  $u_{\infty}(\cdot)$.} 
\label{fig:ofc}

\end{center}

\end{figure}

Figure \ref{fig:ofc} shows the feedback control functions (for a
smaller range near $x=0$ than in Figure \ref{fig:pvf}). We can see the
finite marginal utility at zero wealth of the pre-income termination
value function and its explicit dependence on the income stream. We
also note that both the optimal feedback controls are linear in
wealth, as expected from \eqref{eq:con} and \eqref{eq:inv}.


\section{Conclusions}
\label{sec:conclusion}

In this paper we have proven a rigorous duality for an infinite
horizon consumption problem with randomly terminating income, thereby
closing the duality gap that arose in \cite{dv09}, in a Black-Scholes
market. The key property of the finiteness of marginal utility at zero
initial wealth emerged, while all the other main tenets of duality
theory were shown to hold. Our results readily extend to finite
horizon versions of the problem as well.

There are still many aspects of such problems that are not well understood.
One is to rigorously analyse the primal and dual HJB equations for the
pre-income termination problem, for which the primal problem involves
maximising utility from both consumption and inter-temporal
wealth.
The other is to fully understand the relations between the rather
different forms of dual characterisation of utility maximisation
problems with random endowment in the literature. We leave these
questions for future research.

{\footnotesize
\section*{Acknowledgements}

We thank two anonymous referees for constructive comments that
improved the paper. Part of this work was carried out during visits by
the second author to Universit\'e Paris Diderot and \'Ecole
Polytechnique. We thank Huy\^en Pham and Nizar Touzi for hospitality,
and are grateful to them and Chris Rogers for valuable discussions.
The last author was supported in part by the EPSRC (UK)  Grant (EP/V008331/1).
}

{\small
\bibliography{dfocwrti_refs}
\bibliographystyle{apa}
}

\end{document}